\documentclass[preprint,12pt]{elsarticle}

\usepackage{amssymb}
\usepackage{amsmath}
\usepackage[utf8]{inputenc}
\usepackage[T1]{fontenc}

\usepackage{amsthm}
\usepackage{yfonts}
\usepackage{array}
\usepackage{hyperref}

\newtheorem{definition}{Definition}
\newtheorem{theorem}{Theorem}
\newtheorem{lemma}{Lemma}
\newtheorem{proposition}{Proposition}
\newtheorem{corollary}{Corollary}
\newtheorem{claim}{Claim}

\newcommand{\outlen}{{\lceil}\frac{\max\{D\}}{\min\{G\}}{\rceil}}
\newcommand{\tlen}{\frac{|O||D||G|}{|O||D||G|-1}((|O||D||G|)^{\outlen} -1)}
\newcommand{\timedpair}{\{(s_1,\varepsilon),(s_2, \varepsilon)\}}

\usepackage[vlined, ruled]{algorithm2e}
\usepackage{tikz}
\usetikzlibrary{arrows,shapes,automata,backgrounds,calc, shapes.arrows, patterns}

\usepackage{bbm}

\journal{Theoretical Computer Science}

\begin{document}

\begin{frontmatter}

\title{Studying homing and synchronizing sequences for Timed Finite State Machines with output delays}

\author[label1]{Evgenii Vinarskii\corref{mycorrespondingauthor}}
\cortext[mycorrespondingauthor]{Corresponding author}
\author[label2]{Jakub Ruszil}
\author[label2]{Adam Roman}
\author[label1]{Natalia Kushik}
\affiliation[label1]{organization={SAMOVAR, T\'el\'ecom SudParis / Institut Polytechnique de Paris},
             city={Paris},
             country={France}}
\affiliation[label2]{organization={Faculty of Mathematics and Computer Science / Jagiellonian University},
             city={Krakow},
             country={Poland}}

\begin{abstract}
The paper introduces final state identification (synchronizing and homing) sequences for Timed Finite State Machines (TFSMs) with output delays and investigates their properties. We formally define the notions of homing sequences (HSs) and synchronizing sequences (SSs) for these TFSMs and demonstrate that several properties that hold for untimed machines do not necessarily apply to timed ones. Furthermore, we explore the applicability of various approaches for deriving SSs and HSs for Timed FSMs with output delays, such as truncated successor tree-based and FSM abstraction-based methods. Correspondingly, we identify the subclasses of TFSMs for which these approaches can be directly applied and those for which other (original) methods are required. Additionally, we evaluate the complexity of existence check and derivation of (shortest) HSs / SSs for TFSMs with output delays.
\end{abstract}



\end{frontmatter}

\section{Introduction}\label{sec:intro}
Timed FSMs and Timed Automata (TAs) are widely used in analysis and synthesis of real-time reactive systems (see for example, some recent works~\cite{larsen_schedulability_analysis, larsen_synthesis,larsen_performance_analysis}). Quite often, the current state of the system under analysis is unknown, and in this case, it is necessary to first bring the system to a known, stable state to further continue its analysis (e.g., apply the test cases, verify certain properties, etc.). Being fundamental in the theory of FAs and FSMs, synchronizing and homing sequences~\cite{Sandberg2005} allow the system to be set to a known state. Indeed, SSs and HSs serve as a base for `gedanken'~\cite{fsm} synchronizing and homing experiments, designed to identify a machine's final state, i.e., the state after the sequence has been applied. In the context of FAs and FSMs, a synchronizing experiment involves applying an SS to bring the machine to a known state, regardless of its initial state. A homing experiment, on the contrary, involves not only applying the sequence but also observing the corresponding output response on it, to determine the final state. 

For real-time systems, modeled with Timed FSMs~\cite{vinarskii_pssv}, the concept of the `gedanken' experiments should be extended to consider the timed aspects for both inputs and outputs. These steps have already been taken in the literature; for example, in~\cite{manuel14} and~\cite{tvardovskii_hs_for_tfsm_with_timed_guards}, the authors extend the classical experiment by associating inputs (or events) with timestamps. In this paper, we add an additional observation point at the output channel, capturing both inputs and outputs with precise timestamps. In this setting, inputs are applied and outputs are produced with these timestamps, and that sometimes increases the complexity of the final state identification problem. Key research challenges in the domain include the decidability and computational complexity of the existence check of SSs and HSs, as well as deriving, whenever possible, shortest sequences of this kind.

The derivation of SSs and HSs for classical FAs and FSMs has been extensively studied and summarized by S. Sandberg~\cite{Sandberg2005}, along with algorithms for deriving these sequences. These algorithms can be grouped by: (i) iterative approaches which construct SS (or HS) for the entire machine by combining those derived for every pair of states~\cite{eppstein},~\cite{natarajan},~\cite{Sandberg2005}~\cite{cerny}, (ii) successor tree-based approaches~\cite{kohavi}, which operate on the sets of states by iteratively splitting or merging them at each step, and (iii) solver-based approaches (\cite{volkov19} and~\cite{qbf_homing}) that formulate the conditions for a machine to be synchronized (homed) using a satisfiable formula. For complete deterministic FSMs, checking the existence of HSs (SSs) is relatively straightforward and can be performed in polynomial time, while the problem of the derivation of a shortest HS and SS is NP-hard~\cite{Sandberg2005}. For non-deterministic but observable FSMs and non-deterministic and partial FAs the problem becomes PSPACE-hard~\cite{kushik_observable}, \cite{Ito}, \cite{volkov19}.

This paper explores various strategies for deriving SSs and HSs for the TFSM with output delays considered in~\cite{vinarskii_pssv} and~\cite{enase23}. In the TFSM with output delays, each transition is associated with a timed guard (interval) and an output delay. A transition is executed only when the applied timed input satisfies the corresponding timed guard, while the output delay specifies the time needed to produce the output after the transition execution. The latter allows the TFSM to accept the next input even if the outputs to the previous inputs are still pending. Such TFSM behavior therefore takes into account possible concurrent procedures that are executed in parallel for computing and producing the proper outputs. Therefore, the TFSM of interest has the following features: (i) it operates as a timed input/output automaton, allowing inputs to be accepted without waiting for their corresponding outputs, and (ii) it preserves the FSM property that the number of inputs is equal to the number of outputs in the corresponding response. This affects the properties of SSs and HSs for the TFSM of interest. For example, we show that for these TFSMs, not every prolongation of a homing sequence is a homing sequence. The latter does not allow to ``directly'' adapt the iterative approaches for deriving an HS for the TFSM with output delays. Despite this, we demonstrate that so-called non-integer timed input sequences exist, for which the prolongation on both sides preserves the homing property. Thus, we adapt the truncated successor tree approach~\cite{kushik_complexity} to derive a shortest HS (SS) for the TFSM of interest.

Another approach for deriving SSs and HSs for a timed machine is to reduce the problem to the derivation of SSs and HSs for an untimed abstraction of the timed machine. This reduction is typically achieved through the construction of region automaton / FSM (see for example~\cite{alur},\cite{uppaal},\cite{doyen2014},\cite{tvardovskii_hs_for_tfsm_with_timed_guards},\cite{tvardovskii_hs_for_tfsm_with_timeouts}). The authors in~\cite{tvardovskii_hs_for_tfsm_with_timed_guards} and~\cite{tvardovskii_hs_for_tfsm_with_timeouts} have proven that this reduction can be effectively applied to Timed FSMs with timed guards~\cite{tvardovskii_hs_for_tfsm_with_timed_guards} and Timed FSMs with timeouts~\cite{tvardovskii_hs_for_tfsm_with_timeouts}. However, the same reduction cannot be directly adapted for the Timed FSMs with output delays considered in this paper. To address this, we propose modifying the FSM abstraction introduced in~\cite{tvardovskii_hs_for_tfsm_with_timeouts} to facilitate the derivation of SSs and HSs. The latter motivates us to study the properties of the correspondence between SSs and HSs for the TFSMs with output delays and their untimed abstractions. It also involves comparing the complexity of the existence check of HSs (SSs) and the derivation of a shortest HS (SS) for both, timed and untimed machines.

The main contributions of this paper are the following: (i) introduction and definition of homing and synchronizing sequences for TFSMs with output delays, (ii) HS and SS existence check and derivation strategies for TFSMs, together with the relevant complexity analysis, and (iii) study of various TFSM classes for which successor-tree based approaches and/or FSM abstractions are applicable or not, in the context of the final state identification.

The structure of the paper is as follows. Section~\ref{sec:background} presents the necessary background. Section~\ref{sec:hs_and_ss} introduces the synchronizing and homing sequences for TFSMs with their properties and adjusts a truncated successor tree for their derivation. Section~\ref{sec:fsm_abstractions} discusses the possibilities of different FSM abstractions to derive SS and HS for the TFSM. Section~\ref{sec:conclusion} concludes the paper.

\section{Background}\label{sec:background}

\subsection{Finite State Machines, homing and synchronizing sequences}

A \textit{Finite State Machine} (FSM)~\cite{fsm} is defined as a tuple $\mathcal{M} = (S, I, O, h_S)$ where $S$ is a finite non-empty set of states, $I$ ($O$) is a finite non-empty input (output) alphabet and $h_S \subseteq S \times I \times O \times S$ is a transition relation. We say that $\mathcal{M}$ is \textit{non-deterministic} if for some pair $(s, i) \in S \times I$, there exist at least two different pairs $(o',s'), (o'',s'') \in O \times S$ such that $(s, i, o', s') \in h_S$ and $(s, i, o'', s'') \in h_S$; otherwise, the FSM is \textit{deterministic}. FSM $\mathcal{M}$ is \textit{complete} if the transition relation is defined for each state/input pair $(s, i) \in S \times I$; otherwise, the FSM is \textit{partial}. If for every two transitions $(s, i, o, s_1), (s, i, o, s_2) \in h_S$ it holds that $s_1 = s_2$, then $\mathcal{M}$ is \textit{observable}, otherwise $\mathcal{M}$ is \textit{non-observable}. Consider FSM $\mathcal{M}_1$ (Fig.~\ref{fig:spec_fsm}), $\mathcal{M}_1$ is complete non-deterministic and non-observable since at state $s_1$ under input $i_2$ transitions $(s_1, i_2, o_2, s_0)$ and $(s_1, i_2, o_2, s_1)$ to two distinct states are defined. In order to introduce final state identification sequences, it is convenient to utilize functions $next\_state_{\mathcal{M}}: S \times I^* \times O^* \rightarrow 2^S$ and $out_{\mathcal{M}}: S \times I^* \rightarrow 2^{O^*}$ together with the transition relation. Given states $s$ and $s'$ of $S$, an input sequence $\alpha = i_1i_2 \dots i_n \in I^*$ and an output sequence $\beta = o_1o_2 \dots o_n \in O^*$, we say that $\alpha / \beta$ \textit{brings} FSM $\mathcal{M}$ from state $s$ to state $s'$ if there exist states $s_1 = s, s_2, \dots, s_n, s_{n+1} = s'$ such that $(s_j, i_j, o_j, s_{j+1}) \in h_S$, for $j \in \{ 1, \dots, n \}$. At the same time, function $next\_state^{nd}_{\mathcal{M}}: S \times I^* \rightarrow 2^S$ is defined as follows: $s' \in next\_state^{nd}_{\mathcal{M}}(s, \alpha)$ for $s \in S$ and $\alpha \in I^{*}$ if and only if there exists $\beta \in O^{*}$ such that $s' \in next\_state_{\mathcal{M}}(s, \alpha, \beta)$. The set of all output sequences that $\mathcal{M}$ can produce at state $s$ in response to $\alpha$ is denoted as $out_{\mathcal{M}}(s, \alpha)$. 

Given a complete observable FSM $\mathcal{M}$.

\begin{definition}
\label{def:ms_fsm}
    Given $s, s' \in S$, an input sequence $\alpha \in I^*$ is merging for states $s$ and $s'$ if $next\_state^{nd}_{\mathcal{M}}(s, \alpha) = next\_state^{nd}_{\mathcal{M}}(s', \alpha) = \{ s'' \}$ for $s'' \in S$.
\end{definition}

\begin{definition}
\label{def:ss_fsm}
    An input sequence $\alpha \in I^*$ is synchronizing for $\mathcal{M}$ if $\exists \bar{s} \in S$ such that for every $s \in S$ we have $next\_state^{nd}_{\mathcal{M}}(s,\alpha) = \{\bar{s}\}$. 
\end{definition}

\begin{definition}
\label{def:hs_fsm}
    An input sequence $\alpha \in I^*$ is homing for $\mathcal{M}$ if for each state pair $\{s_1, s_2\}$ of $\mathcal{M}$ the following holds: $\forall \beta \in out_{\mathcal{M}}(s_1, \alpha) \cap out_{\mathcal{M}}(s_2, \alpha)$ it holds that $next\_state_{\mathcal{M}}(s_1, \alpha, \beta) = next\_state_{\mathcal{M}}(s_2, \alpha, \beta) = \{ \bar{s} \}$ for some $\bar{s} \in S$.
\end{definition}

Note that, according to Definitions~\ref{def:ss_fsm} and~\ref{def:hs_fsm}, a synchronizing sequence is always homing, but the converse is not true. The properties of merging, synchronizing, and homing sequences and their relationship have been extensively studied for FAs and FSMs and their different modifications~\cite{Sandberg2005}. It has been shown that a complete deterministic automaton is synchronizing if and only if every pair of distinct states has a merging sequence~\cite{eppstein},~\cite{natarajan} and~\cite{cerny}. This property provides an intuition for deriving, not necessarily shortest, synchronizing sequences by iteratively deriving merging sequences for every pair of states and appending the merging sequences for the successors. The approach relies on the fact that if a sequence homes (or merges) two different states, then any prolongation of this sequence also homes (or merges) these states. In this paper, we investigate whether similar properties of final state identification sequences hold for Timed FSMs with output delays and how the SS and HS derivation may differ with respect to the timed aspects.

\subsection{Timed FSMs with output delays \& related notions}

As mentioned in Section~\ref{sec:intro}, the behavior of a TFSM with output delays at a current state depends on the input, the time when this input is applied and the time required to handle it. These aspects of a real-time behavior can be formalized with \textit{timestamps}, \textit{timed guards}, and \textit{output delays} (or just \textit{delays}). A \textit{timestamp} $t \in \mathbb{R}^{+}_{0}$ specifies a time instance at which a real-time system can receive an input or generate an output. A timed guard $g=[u,v)$, for $u, v \in \mathbb{N}^{+}_0$ and $0 \leq u < v$, indicates the time \emph{window} during which a system transition is enabled to handle the input\footnote{We also consider point timed guards (intervals) $[u,u]$ in Section~\ref{sec:fsm_abstractions}.}. A delay $d \in \mathbb{N}^{+}$ specifies the time required to generate an output after an input has been received. Given finite non-empty input (output) alphabet $I$ ($O$) and a timestamp $t$, we say that a \textit{timed input} (\textit{output}) is a pair $(a,t)$ where $a \in I$ ($a \in O$). A \textit{timed input} (\textit{output}) \textit{sequence} $\alpha = (a_1, t_1) \dots (a_n, t_n)$ is a finite sequence of timed inputs (outputs) where sequence $t_1 \dots t_n$ is non-decreasing.

\begin{definition} A TFSM $\mathcal{S}$ with output delays, over finite $I$ and $O$, is a tuple $(S, I, O, G, D, h_S)$, where $S$, $G$ and $D$ are finite non-empty sets of states, timed guards and output delays, respectively, while $h_S \subseteq S \times I \times G \times O \times D \times S$ is a transition relation.
\end{definition}

\begin{figure}[ht]
\centering
\begin{minipage}{.42\textwidth}
  \centering
  \begin{tikzpicture}[>=stealth',node distance=3.1cm,semithick,auto, scale=1.0, every node/.style={scale=0.7}]
				\node[state, minimum size=1cm]	(s0)
				    {$s_0$};
                \node[state, right of=s0,minimum size=1cm]	(s1)
				    {$s_1$};
				\node[state, right of=s1,minimum size=1cm]	(s2)
				    {$s_2$};

				{
				\path[->] 
				
                (s0) edge [loop above, align=left] node {
                        $i_2/o_3$ \\
                        $i_2/o_1$
                        } (s0)
                (s0) edge [bend left, above, align=left] node {
                        $i_1/o_1$
                        } (s1)
                (s1) edge [bend left, above, align=left] node {
                        $i_1/o_2$
                        } (s2)
                (s2) edge [bend left, below, align=left] node {
                        $i_2/o_1$
                        } (s1)
                (s1) edge [bend left, below, align=left] node {
                        $i_2/o_2$
                        } (s0)
                (s1) edge [loop above, align=left] node {
                        $i_2/o_2$
                        } (s1)
                (s2) edge [loop above, align=left] node {
                        $i_1/o_3$ \\
                        $i_2/o_3$
                        } (s2)
					;
				}
\end{tikzpicture}

  \caption{FSM $\mathcal{M}_1$}
  \label{fig:spec_fsm}
\end{minipage}%
\hfill
\begin{minipage}{.55\textwidth}
  \centering
  \begin{tikzpicture}[>=stealth',node distance=3.7cm,semithick,auto, scale=1.0, every node/.style={scale=0.7}]
			    \node[state, minimum size=1cm]	(s0)
				    {$s_0$};
                \node[state, right of=s0,minimum size=1cm]	(s1)
				    {$s_1$};
                \node[state, right of=s1,minimum size=1cm]	(s2)
				    {$s_2$};

				{
				\path[->] 
				
                (s0) edge [loop above, align=left] node {
                        $i_2,[1,3)/o_3,1$ \\
                        $i_2,[3,6)/o_1,4$
                        } (s0)
                (s0) edge [bend left, above, align=left] node {
                        $i_1,[1, 3) / o_1, 4$
                        } (s1)
                (s1) edge [bend left, above, align=left] node {
                        $i_1,[1,3)/o_2,3$
                        } (s2)
                (s2) edge [bend left, below, align=left] node {
                        $i_2,[1,3)/o_1,2$ \\
                        $i_2,[5,6)/o_1,2$
                        } (s1)
                (s1) edge [bend left, below, align=left] node {
                        $i_2,[1,3)/o_2, 1$ \\
                        $i_2,[4,6)/o_2, 1$ \\
                        } (s0)
                (s1) edge [loop above, align=left] node {
                        $i_2,[3,4)/o_2,1$
                        } (s1)
                (s2) edge [loop above, align=left] node {
                        $i_1,[1,3)/o_3,1$ \\
                        $i_2,[3,5)/o_3,3$
                        } (s2)
					;
				}
\end{tikzpicture}

  \caption{TFSM $\mathcal{S}_1$}
  \label{fig:spec_tfsm}
\end{minipage}
\end{figure}

A transition $(s, i, g, o, d, s') \in h_S$ means that, receiving input $i$ at state $s$ after $t\in g$ time units $\mathcal{S}$ moves to state $s'$, producing output $o$ after $d$ time units (written as $s\stackrel{i,g/o,d}{\longrightarrow}s'$). Given $t_0 = 0$, $s \in S$ and $\alpha = (i_1, t_1) (i_2, t_2) \dots (i_n, t_n)$, we say that $\alpha$ \textit{induces} a sequence of transitions $tr = s \stackrel{i_1,g_1/o_1,d_1}{\longrightarrow} s_1 \stackrel{i_2,g_2/o_2,d_2}{\longrightarrow} \dots \stackrel{i_n,g_n/o_n,d_n}{\longrightarrow} s_n$ from state $s$ if $t_1 - t_0 \in g_1$, $t_2 - t_1 \in g_2$, \dots, $t_n - t_{n-1} \in g_n$. The set of all timed input sequences that induce at least one sequence of transitions from state $s$ is denoted as $Dom_{\mathcal{S}}(s)$. If for every pair of transitions $s\stackrel{i,g_1/o_1,d_1}{\longrightarrow}s_1$ and $s\stackrel{i,g_2/o_2,d_2}{\longrightarrow}s_2$ in $\mathcal{S}$ it holds that $g_1 \cap g_2 = \emptyset$, then $\mathcal{S}$ is \textit{deterministic}. In this paper, we study the final state identification for deterministic timed machine.

The \textit{run} induced by $\alpha$ is denoted as $r = s_0 \xrightarrow[o_1, \tau_1]{i_1, t_1} s_1 \xrightarrow[o_2, \tau_2]{i_2, t_2} \dots \xrightarrow[o_n, \tau_n]{i_n, t_n} s_n$, where $\tau_1 = t_1 + d_1$, $\tau_2 = t_2 + d_2$, \dots, $\tau_n = t_n + d_n$. We denote as $s \xrightarrow[o, \tau]{i, t} s'$ the fact that being at $s$ and receiving $(i, t)$, the machine immediately moves to $s'$ and produces $o$ at $\tau = t + d$. Unlike TFSMs considered in~\cite{bresolin_2014}, the next timed input can be applied before the machine has produced outputs to the previous inputs. Consequently, the sequence of timed outputs $r\downarrow_{O}=(o_1,\tau_1)(o_2,\tau_2)\dots(o_n,\tau_n)$ may not represent the exact output response of $\mathcal{S}$ to $\alpha$. Specifically, there can exist indices $\ell, k \in \{ 1, \dots, n \}$ such that $\ell < k$ but $\tau_{\ell} > \tau_k$, indicating that output $o_{\ell}$ is produced after output $o_k$. Additionally, $\tau_{\ell}$ can be equal to $\tau_k$, in this case the outputs can occur simultaneously. Therefore, the \textit{timed output response} of $\mathcal{S}$ to $\alpha$, denoted as $timed\_out_{\mathcal{S}}(s, \alpha)$, is defined as the set of all possible permutations $(o_{j_1}, \tau_{j_1}) (o_{j_2}, \tau_{j_2}) \dots (o_{j_n}, \tau_{j_n})$ of $r\downarrow_{O}$ such that $\tau_{j_1} \leq \tau_{j_2} \leq \dots \leq \tau_{j_n}$.

As an example, consider TFSM $\mathcal{S}_1$ (Fig.~\ref{fig:spec_tfsm}) and $\alpha_1=(i_1,2)(i_2,4)(i_2,5)$. If $\mathcal{S}_1$ is at $s_0$, then $i_1$ can be processed when received at $t_1 \in [1,3)$. Thus, $s_0\stackrel{i_1,[1,3)/o_1,4}{\longrightarrow}s_1$ is enabled for $(i_1,2)$ and $\mathcal{S}_1$ moves to $s_1$, while $o_1$ will be produced at $2+4=6$ time units. Since $4-2 \in [1,3)$, $s_1\stackrel{i_2,[1,3)/o_2,1}{\longrightarrow}s_0$ is enabled for $(i_2,4)$, $\mathcal{S}_1$ moves to $s_0$, while $o_2$ will be produced at $4+1=5$ time units. Similarly, $s_0\stackrel{i_2,[1,3)/o_3,1}{\longrightarrow}s_0$ is enabled for $(i_2,5)$ and $o_3$ will be produced at $5+1=6$ time units. So $\alpha_1$ is enabled for $\mathcal{S}_1$ and the run induced by $\alpha_1$ at state $s_0$ is $r = s_0 \xrightarrow[o_1,6]{i_1,2} s_1 \xrightarrow[o_2,5]{i_2,4} s_0 \xrightarrow[o_3,6]{i_2,5} s_0$. Note that $i_1$ was applied before $i_2$, but the timestamp of output $o_2$ is less than that of output $o_1$, that is, $o_2$ will be produced before $o_1$. At the same time, the timestamps of outputs $o_1$ and $o_3$ are the same. Thus, $r\downarrow_{O} = (o_1,6)(o_2,5)(o_3,6)$ is not a timed output sequence, and $\mathcal{S}_1$ produces either $\beta_1 = (o_2,5)(o_1,6)(o_3,6)$ or $\beta_2 = (o_2,5)(o_3,6)(o_1,6)$, i.e., due to the competition, $o_1$ and $o_3$ can be produced at the same time instance. Therefore, the response of $\mathcal{S}_1$ to $\alpha_1$ is $timed\_out_{\mathcal{S}_1}(s_0, \alpha_1) = \{ \beta_1, \beta_2 \}$.

This example showcases the difference between the TFSM with output delays and the timed machines considered in~\cite{bresolin_2014}. Consider the execution $s\stackrel{i_1,g_1/o_1,d_1}{\rightarrow}s_1\stackrel{i_2,g_2/o_2,d_2}{\rightarrow}s_2$. In the TFSM with output delays, when the machine is at $s$ and receives $(i_1,t)$, it immediately moves to $s_1$, while starting procedure $f_1$ to process $i_1$ and to generate $o_1$. The execution of $f_1$ requires $d_1$ time units. A subsequent input $i_2$ may be applied even before output $o_2$ has been produced. In that case, the TFSM immediately moves to $s_2$ and starts procedure $f_2$ to produce $o_2$, when $f_1$ is still pending.

Together with the transition relation for a deterministic TFSM $\mathcal{S} = (S, I, O, G,$ $ D, h_S)$, we define the \textit{transition function} $next\_state_{\mathcal{S}}:\ (S \cup \{\bot\}) \times (I \times \mathbb{R}^{+}) \rightarrow (S \cup \{\bot\})$, where $\bot \not\in S$ in the following way: if there exists a transition $s\stackrel{i,g/o,d}{\longrightarrow}s'$ such that $t \in g$, then $next\_state_{\mathcal{S}}(s, (i, t)) = s'$, otherwise $next\_state_{\mathcal{S}}(s, (i, t)) = \bot$. As for classical FSMs, $next\_state_{\mathcal{S}}$ function can be extended to timed input sequences. Let $\alpha = (i_1, t_1) \dots (i_n, t_n)$ be a timed input sequence and $\alpha' = \alpha (i_{n+1},t_{n+1}) = (i_1, t_1) \dots (i_n, t_n) (i_{n+1},t_{n+1})$ be a prolongation of $\alpha$; if $next\_state_{\mathcal{S}}(s, \alpha) = \bot$, then $next\_state_{\mathcal{S}}(s, \alpha') = \bot$, otherwise it is defined as 
$$
next\_state_{\mathcal{S}}(s, \alpha') = next\_state_{\mathcal{S}}(next\_state_{\mathcal{S}}(s, \alpha), (i_{n+1}, t_{n+1}-t_n)).
$$
\begin{proposition}\label{prop:next_state}
If $\mathcal{S}$ is a deterministic TFSM with output delays, $s$ is a state of $\mathcal{S}$ and $\alpha \in Dom_{\mathcal{S}}(s)$, then $|next\_state_{\mathcal{S}}(s, \alpha)| = 1$.
\end{proposition}

Another important property of an FSM is to be complete~\cite{fsm}; for TFSMs \emph{completeness} can be defined in various ways. We say that $\mathcal{S}$ is \textit{strongly-complete} if $next\_state_{\mathcal{S}}$ function is total, i.e., for every state $s$ and for every timed input $(i, t)$ there exists a transition $s \stackrel{i,g/o,d}{\rightarrow}s' \in h_S$. We say that $\mathcal{S}$ is \textit{weakly-complete} if for every pair of states $s, s'$ their domains are equal, i.e., $Dom_{\mathcal{S}}(s) = Dom_{\mathcal{S}}(s')$, otherwise $\mathcal{S}$ is \textit{partial}. As an example, consider again TFSM $\mathcal{S}_1$ (Fig.~\ref{fig:spec_tfsm}), $[1, 3)$ is the single timed guard defined for input $i_1$ at states $s_0, s_1$ and $s_2$; at the same time, timed guards $[1, 3)$ and $[3,4),[3,5),[5,6),[4,6),[3,6)$ are defined at states $s_0, s_1$ and $s_2$ for $i_2$ and cover $[1, 6)$ for all states, therefore, $\mathcal{S}_1$ is weakly-complete, but not strongly-complete. Note that FSM $\mathcal{M}_1$ (Fig.~\ref{fig:spec_fsm}) is derived by erasing all timed guards and output delays from TFSM $\mathcal{S}_1$. We further refer to such FSMs as \emph{FSM-abstractions} and study their properties (see Section~\ref{subsec:1_fsm_abstraction}).

\section{Homing and synchronizing sequences for TFSMs with output delays}\label{sec:hs_and_ss}
\subsection{Final state identification sequences and their properties}\label{sec:properties_ss_hs}

Let $\mathcal{S} = (S, I, O, G, D, h_S)$ be a TFSM with output delays, we introduce the following definitions\footnote{Since we consider deterministic TFSMs, function $next\_state_{\mathcal{S}}$ returns at most one state.}.

\begin{definition}
Given $s, s' \in S$, a timed input sequence $\alpha$ is merging for states $s$ and $s'$ if $next\_state_{\mathcal{S}}(s, \alpha) = next\_state_{\mathcal{S}}(s', \alpha) \neq \bot$.
\end{definition}

\begin{definition}\label{def:ss}
A timed input sequence $\alpha$ is \textit{synchronizing} for $\mathcal{S}$ if $\exists \bar{s} \in S$ such that for every $s \in S$ we have $next\_state_{\mathcal{S}}(s,\alpha) = \bar{s}$.
\end{definition}

Consider the behavior of TFSM $\mathcal{S}_1$ (Fig.~\ref{fig:spec_tfsm}) after applying the timed input sequence $\gamma=(i_1,2)(i_1,4)(i_1,6)$. Since $next\_state_{\mathcal{S}_1}(s_0,\gamma) = s_2$,\\ $next\_state_{\mathcal{S}_1}(s_1,\gamma) = s_2$ and $next\_state_{\mathcal{S}_1}(s_2,\gamma) = s_2$, $\gamma$ is a synchronizing sequence for $\mathcal{S}_1$. Since the next state of the deterministic TFSM is uniquely defined (Proposition~\ref{prop:next_state}), we conclude that the existence of a merging sequence for each state pair implies the existence of an SS for the TFSM.

\begin{definition}
\label{def:hs}
A timed input sequence $\alpha$ is \textit{homing} for $\mathcal{S}$ if for each state pair $\{s_1, s_2\}$ of $\mathcal{S}$ the following holds: $timed\_out_{\mathcal{S}}(s_1, \alpha) = timed\_out_{\mathcal{S}}(s_2, \alpha)$ implies that $next\_state_{\mathcal{S}}(s_1, \alpha)$ $ = next\_state_{\mathcal{S}}(s_2, \alpha) \neq \bot$.
\end{definition}

According to the definition of synchronizing and homing sequences and due to Proposition~\ref{prop:next_state}, we conclude that any synchronizing sequence remains a homing sequence for a TFSM with output delays. In the context of complete deterministic FSMs, any right/left prolongation of a homing sequence remains homing (see Section~\ref{sec:background}). However, we show that it is not the case even for deterministic Timed FSMs with output delays. Consider the behavior of $\mathcal{S}_1$ (Fig.~\ref{fig:spec_tfsm}) after applying $\alpha_1=(i_1,2)$. Since $timed\_out_{\mathcal{S}_1}(s_0,\alpha_1)=\{(o_1,6)\}$, $timed\_out_{\mathcal{S}_1}(s_1,\alpha_1)=\{(o_2,5)\}$ and $timed\_out_{\mathcal{S}_1}(s_2,\alpha_1)=\{(o_3,3)\}$, we conclude, that $\alpha_1$ is an HS for $\mathcal{S}_1$. Now, consider the behavior of $\mathcal{S}_1$ on $\alpha_1'=(i_1,2)(i_2,4)$ which is the right prolongation of $\alpha_1$. Since $next\_state_{\mathcal{S}_1}(s_0, \alpha_1')=s_0$ and $next\_state_{\mathcal{S}_1}(s_1, \alpha_1')=s_1$, $next\_state_{\mathcal{S}_1}(s_0,\alpha_1') \neq next\_state_{\mathcal{S}_1}(s_1,\alpha_1')$. At the same time, $\alpha_1'$ induces run $r_{s_0} = s_0 \xrightarrow[o_1,6]{i_1,2} s_1 \xrightarrow[o_2,5]{i_2,4} s_0$ at state $s_0$ and induces run $r_{s_1} = s_1 \xrightarrow[o_2,5]{i_1,2} s_2 \xrightarrow[o_1,6]{i_2,4} s_1$ at state $s_1$; therefore $timed\_out_{\mathcal{S}_1}(s_0,\alpha_1')=timed\_out_{\mathcal{S}_1}(s_1,\alpha_1')=\{(o_2,5)(o_1,6)\}$. Thus, $\alpha_1'$ is not homing for $\mathcal{S}_1$. The fact that the right prolongation of a homing sequence might stop being homing is also true for non-observable FSMs~\cite{qbf_homing}. However, the left prolongation of a homing sequence remains a homing sequence for complete non-observable FSMs; indeed, any finite amount of inputs can be added before a homing sequence without destroying this property. At the same time, this is not the case for the TFSM with output delays. As an example, consider again TFSM $\mathcal{S}_1$ (Fig.~\ref{fig:spec_tfsm}), $\alpha_2 = (i_2,2)$ is an HS for $\mathcal{S}_1$ since $timed\_out_{\mathcal{S}_1}(s_0, \alpha_2) = \{(o_3, 3)\}$, $timed\_out_{\mathcal{S}_1}(s_1, \alpha_2) = \{(o_2, 3)\}$ and $timed\_out_{\mathcal{S}_1}(s_2, \alpha_2) = \{(o_1, 4)\}$. Now consider $\alpha'_2 = (i_2,4)(i_2,6)$ which is the left prolongation of $\alpha_2$, $\alpha'_2$ induces run $r'_{s_0} = s_0 \xrightarrow[o_1,8]{i_2,4} s_0 \xrightarrow[o_3,7]{i_2,6} s_0$ at state $s_0$ and induces run $r'_{s_2} = s_2 \xrightarrow[o_3,7]{i_2,4} s_2 \xrightarrow[o_1,8]{i_2,6} s_1$ at state $s_2$; therefore $timed\_out_{\mathcal{S}_1}(s_0,\alpha'_2)=timed\_out_{\mathcal{S}_1}(s_2,\alpha'_2)=\{(o_3,7)(o_1,8)\}$ while $next\_state_{\mathcal{S}_1}(s_0,\alpha'_2) \neq next\_state_{\mathcal{S}_1}(s_2,\alpha'_2)$. Thus, $\alpha'_2$ is not a homing sequence.

These examples demonstrate that the assumption that if two distinct states have been homed at some point, they will remain homed for any prolongation, is not necessarily true for TFSMs. In fact, such behavior can happen only for very specific timed input sequences.

\begin{lemma}\label{lemma:prolongation_not_hs}
Let $\mathcal{S}$ be a TFSM with output delays and $\alpha$ be a homing sequence for $\mathcal{S}$. If a prolongation $\alpha' = (i_1,t_1)\dots(i_n,t_n)$ of $\alpha$ is not a homing sequence, then there exist $\ell, r \in \{ 1, \dots, n \}$ such that $t_r - t_{\ell} \in \mathbb{N}^{+}_0$.
\end{lemma}
\begin{proof}
Since $\alpha'$ is not an HS for $\mathcal{S}$, there exist states $s$ and $s'$ of $\mathcal{S}$ such that $timed\_out_{\mathcal{S}}(s, \alpha') = timed\_out_{\mathcal{S}}(s', \alpha')$ while $next\_state_{\mathcal{S}}(s, \alpha') \neq next\_state_{\mathcal{S}}(s', \alpha')$. Since $\alpha$ is a proper prefix of $\alpha'$, we conclude that $next\_state_{\mathcal{S}}(s, \alpha) \neq next\_state_{\mathcal{S}}(s', \alpha)$. Therefore, due to the fact that $\alpha$ is a homing sequence, $timed\_out_{\mathcal{S}}(s, \alpha) \neq timed\_out_{\mathcal{S}}(s', \alpha)$. 
$\alpha'$ induces run $r = s \xrightarrow[o_1, t_1 + d_1]{i_1, t_1} \dots s_{m-1} \xrightarrow[o_m, t_m + d_m]{i_m, t_m} s_m \dots s_{n-1} \xrightarrow[o_n, t_n + d_n]{i_n, t_n} s_n$ at state $s$ and run $r' = s' \xrightarrow[o'_1, t_1 + d'_1]{i_1, t_1} \dots s'_{m-1} \xrightarrow[o'_m, t_m + d'_m]{i_m, t_m} s'_m \dots s'_{n-1} \xrightarrow[o'_n, t_n + d'_n]{i_n, t_n} s'_n$ at state $s'$.

Since $r\downarrow_{O} = (o_1, t_1 + d_1)\dots(o_n, t_n + d_n)$, $timed\_out_{\mathcal{S}}(s, \alpha') = \{ (o_{j_1}, t_{j_1} + d_{j_1})\dots(o_{j_n}, t_{j_n} + d_{j_n}) \}$ is a permutation $j$ of $r\downarrow_{O}$ such that $t_{j_1} + d_{j_1} \leq \dots \leq t_{j_n} + d_{j_n}$. Similarly, since $r'\downarrow_{O} = (o'_1, t_1 + d'_1)\dots(o'_n, t_n + d'_n)$, $timed\_out_{\mathcal{S}}(s', \alpha')$ $ = \{(o'_{k_1}, t_{k_1} + d'_{k_1})\dots(o'_{k_n}, t_{k_n} + d'_{k_n})\}$ is a permutation $k$ of $r'\downarrow_{O}$ such that $t'_{k_1} + d'_{k_1} \leq \dots \leq t'_{k_n} + d'_{k_n}$.

As $timed\_out_{\mathcal{S}}(s, \alpha') = timed\_out_{\mathcal{S}}(s', \alpha')$, it holds that  $o_{j_1} = o_{k_1}$, \dots, $o_{j_n} = o_{k_n}$ and $t_{j_1} + d_{j_1} = t_{k_1} + d'_{k_1}$, \dots, $t_{j_n} + d_{j_n} = t_{k_n} + d'_{k_n}$. Since $timed\_out_{\mathcal{S}}(s, \alpha) \neq timed\_out_{\mathcal{S}}(s', \alpha)$, we conclude that $j$ and $k$ are different permutations. The latter implies $|t_{j_1} - t_{k_1}| = |d'_{k_1} - d_{j_1}| \in \mathbb{N}^+_0$, $|t_{j_2} - t_{k_2}| = |d'_{k_2} - d_{j_2}| \in \mathbb{N}^+_0$ \dots, $|t_{j_n} - t_{k_n}| = |d'_{k_n} - d_{j_n}| \in \mathbb{N}^+_0$. Therefore, there exist $r, \ell \in \{1,\dots,n \}:\ |t_{r} - t_{\ell}| \in \mathbb{N}^+_0$.
\end{proof}

Therefore, if a right/left prolongation of a homing sequence stops being homing, then there exist at least two timed inputs such that the difference between their timestamps is integer. Consider $\alpha = (i_1,t_1)\dots(i_n,t_n)$ with the following property: the difference between any two timestamps is non-integer, i.e., $t_j-t_k\not\in \mathbb{N}^{+}_0$ for every $j, k \in \{ 1, \dots, n \}$, $j\neq k$; we refer to such timed input sequence as \textit{non-integer}. The following corollary holds.
\begin{corollary}\label{corollary:non_integral_seq}
If $\alpha$ is a homing sequence for a TFSM $\mathcal{S}$, then any non-integer right/left prolongation of $\alpha$ remains homing.
\end{corollary}

\subsection{Deriving a shortest SS/HS for a TFSM: truncated successor tree approach}\label{subsec:notation}

Let $\mathcal{S}$ be a TFSM, $s$ be a state of $\mathcal{S}$, $\alpha$ and $\alpha'$ be timed input sequences inducing the same sequence of transitions from state $s$. We say that such sequences $\alpha$ and $\alpha'$ are \textit{equivalent for state $s$}, denoted as $\alpha \sim_s \alpha'$. Similarly, if $\alpha$ and $\alpha'$ are equivalent for every state, we say that $\alpha$ and $\alpha'$ are \textit{equivalent for TFSM $\mathcal{S}$}, denoted as $\alpha \sim_{\mathcal{S}} \alpha'$. The following theorem holds.
\begin{theorem}\label{theorem:congruent_sequences}
Given non-integer timed input sequences $\alpha = (i_1,t_1)\dots(i_n,t_n)$ and $\alpha' = (i_1,t'_1)\dots(i_n,t'_n)$, if $\alpha \sim_{\mathcal{S}} \alpha'$, then $\alpha$ is synchronizing (homing) for $\mathcal{S}$ if and only if $\alpha'$ is synchronizing (homing) for $\mathcal{S}$.
\end{theorem}
\begin{proof}
1. Suppose that $\alpha$ is an SS for $\mathcal{S}$, but $\alpha'$ is not an SS for $\mathcal{S}$. There exist states $s, s' \in S$ such that $next\_state_{\mathcal{S}}(s, \alpha) = next\_state_{\mathcal{S}}(s', \alpha)$, but, $next\_state_{\mathcal{S}}(s, \alpha') \neq next\_state_{\mathcal{S}}(s', \alpha')$. However, $\alpha \sim_{s} \alpha'$ and $\alpha \sim_{s'} \alpha'$, therefore, $next\_state_{\mathcal{S}}(s, \alpha) = next\_state_{\mathcal{S}}(s, \alpha')$ and $next\_state_{\mathcal{S}}(s', \alpha) = next\_state_{\mathcal{S}}(s', \alpha')$, thus $next\_state_{\mathcal{S}}(s, \alpha') $ $= next\_state_{\mathcal{S}}(s', \alpha')$, it is a contradiction, therefore, $\alpha'$ is also an SS.

2. Suppose that $\alpha$ is an HS for $\mathcal{S}$, but $\alpha'$ is not an HS for $\mathcal{S}$. Then there exist states $s, s' \in S$ such that $timed\_out_{\mathcal{S}}(s, \alpha') = timed\_out_{\mathcal{S}}(s', \alpha')$ and $next\_state_{\mathcal{S}}(s, \alpha') \neq next\_state_{\mathcal{S}}(s', \alpha')$, $\alpha'$ induces run $r = s \xrightarrow[o_1, t'_1 + d_1]{i_1, t'_1} \dots s_{n-1} \xrightarrow[o_n, t'_n + d_n]{i_n, t'_n} s_n$ at state $s$ and run $r' = s' \xrightarrow[o'_1, t'_1 + d'_1]{i_1, t'_1} \dots s'_{n-1} \xrightarrow[o'_n, t'_n + d'_n]{i_n, t'_n} s'_n$ at state $s'$. Since $\alpha'$ is a non-integer timed input sequence, $timed\_out_{\mathcal{S}}(s, \alpha') = timed\_out_{\mathcal{S}}(s', \alpha')$ if and only if $o_1 = o'_1$, \dots, $o_n = o'_n$ and $d_1 = d'_1$, \dots, $d_n = d'_n$ (see the proof of  Lemma~\ref{lemma:prolongation_not_hs}). As $\alpha \sim_{\mathcal{S}} \alpha'$, it holds that $next\_state_{\mathcal{S}}(s, \alpha) \neq next\_state_{\mathcal{S}}(s', \alpha)$ and $timed\_out_{\mathcal{S}}(s, \alpha') \neq timed\_out_{\mathcal{S}}(s', \alpha')$. Moreover, runs induced by $\alpha$ at states $s$ and $s'$ are the following: $r = s \xrightarrow[o_1, t_1 + d_1]{i_1, t_1} \dots \xrightarrow[o_n, t_n + d_n]{i_n, t_n} s_n$ and $r' = s' \xrightarrow[o_1, t_1 + d_1]{i_1, t_1} \dots \xrightarrow[o_n, t_n + d_n]{i_n, t_n} s'_n$. Thus, $timed\_out_{\mathcal{S}}(s, \alpha) = timed\_out_{\mathcal{S}}(s', \alpha)$, it is a contradiction with the fact that $\alpha$ is an HS.
\end{proof}

Theorem~\ref{theorem:congruent_sequences} claims that in order to derive a homing (synchronizing) sequence for a TFSM with output delays, it is sufficient to explore only non-equivalent timed input sequences. Let $s \in S$, $(i, t)$ be a timed input and $(o, d) \in O \times D$; we say that $s' \in S$ is $(i,t)/\{(o,t+d)\}$-\textit{successor} of $s$ if $s' = next\_state_{\mathcal{S}}(s, (i, t))$ and $\{(o, t+d)\} = timed\_out_{\mathcal{S}}(s, (i,t))$, written as $s' = (i,t)/\{(o,t+d)\}$-$succ(s)$. Consider, for example, TFSM $\mathcal{S}_2$ (Fig.~\ref{fig:fsm_A}). Since $\mathcal{S}_2$ has the transition $s_0\stackrel{i_2,[2, 4)/o_2,1}{\longrightarrow}s_1$, we conclude that $s_1 = next\_state_{\mathcal{S}_2}(s_0,(i_2,2.5))$ and $\{(o_2,3.5)\} = timed\_out_{\mathcal{S}_2}(s_0,(i_2,2.5))$, thus, $s_1=(i_2,2.5)/\{(o_2,3.5)\}$-$succ(s_0)$.
Function $(i,t)/\{(o,t+d)\}$-$succ: S \rightarrow S$ can be extended to operate over the subsets of states, i.e., $(i,t)/\{(o,t+d)\}$-$succ: 2^S \rightarrow 2^S$. In particular, let $S_1, S_2$ be subsets of $S$, we say that $S_2 = (i,t)/\{(o,t+d)\}$-$succ(S_1)$ if and only if for every state $s' \in S_2$ there exists a state $s \in S_1$ such that $s' = (i,t)/\{(o,t+d)\}$-$succ(s)$. For example, for TFSM $\mathcal{S}_2$ shown in Fig.~\ref{fig:fsm_A}, $\{ s_0, s_1 \} = (i_2,2.5)/\{(o_2,3.5)\}$-$succ(\{ s_0, s_1, s_2 \})$. 

In this section, we deal only with weakly-complete deterministic TFSMs, i.e., for every two states $s$ and $s'$ of TFSM $\mathcal{S}$ it holds that $Dom_{\mathcal{S}}(s) = Dom_{\mathcal{S}}(s')$. Formally, let $i$ be an input, transitions $s \stackrel{i,g_1/o_1,d_1}{\rightarrow} s_1$, \dots, $s \stackrel{i,g_m/o_m,d_m}{\rightarrow} s_m$ are defined at state $s$ and transitions $s' \stackrel{i,g'_1/o'_1,d'_1}{\rightarrow} s'_1$, \dots, $s' \stackrel{i,g'_k/o'_k,d'_k}{\rightarrow} s'_k$ are defined at state $s'$; the weakly-complete property means that $g_1 \cup \dots \cup g_m = g'_1 \cup \dots \cup g'_k$. For further computation, we denote as $U_i$ and $V_i$ the left and right boundaries of timed interval $g_1 \cup \dots \cup g_m$. Given a weakly-complete TFSM $\mathcal{S}$ with $n$ states, in order to derive an HS for $\mathcal{S}$, we propose to construct a truncated successor tree (TST), such that each shortest HS, up to the equivalence, is contained in it. The derivation of the tree as well as the proper truncating rules are presented in Algorithm~\ref{alg:hs_synthesis}.

\begin{algorithm}[!htb]
\small
\SetKwInOut{Input}{input}\SetKwInOut{Output}{output}
\Input{A weakly-complete TFSM $\mathcal{S} = (S, I, O, G, D, h_S)$ with output delays
}
\Output{Message ``There is no HS for $\mathcal{S}$'' or a shortest HS for $\mathcal{S}$}
\textbf{Step 1.} Derivation of a truncated successor tree for $\mathcal{S}$ \\

The root of the tree is the node labeled with set $S$ while each node of the successor tree is labeled with a set of subsets of states; edges of the tree are labeled with timed inputs. Given a non-terminal node labeled with a set $P$ of subsets of states, at level $j, j \geq 0$, and an input $i$ with the minimal left $U_i$ and maximal right $V_i$ boundaries. There is the edge labeled with timed input $(i, k+2^{-j})$ where $k$ is an every integer such that $k \in [U_i, V_i)$, to node labeled with set $Q$ of subsets of states, at level $j+1$ if and only if $S_Q \in Q$, where $S_Q = (i,k+2^{-j})/\{(o,k+2^{-j}+d)\}$--$succ(S_P)$ for some $S_P \in S$ and some $(o, d) \in O \times D$.

Given a node at level $\ell$ labeled with set $P$, the node is terminal if one of the following conditions hold:

\textbf{Rule 1.} $P$ contains only singletons.

\textbf{Rule 2.} $P$ contains a set $R$ without singletons that labels a node at a level $j, j \leq \ell$.

\textbf{Step 2.} \If {the successor tree has no node truncated using the Rule 1}
{
    \Return {the message ``There is no HS for $\mathcal{S}$''}
}

// There is a node in the derived TST truncated using Rule 1.\\
Choose a node $P$ truncated with Rule 1 with the minimal depth. \\
// The path from the root to node $P$ is labeled with $(i_1, \delta_1) (i_2, \delta_2) \dots (i_{\ell}, \delta_{\ell})$

\Return {$(i_1, \delta_1)(i_2, \delta_1 + \delta_2)\dots(i_{\ell}, \delta_1 + \delta_2 + \dots + \delta_{\ell})$
}
\caption{Deriving a shortest HS for a weakly-complete TFSM with output delays}\label{alg:hs_synthesis}
\end{algorithm}

\begin{theorem}[Correctness of Algorithm~\ref{alg:hs_synthesis}]\label{theorem:correctness_of_hs_synthesis_alg}
A weakly-complete deterministic TFSM $\mathcal{S}$ has a homing sequence if and only if the truncated successor tree derived by Algorithm~\ref{alg:hs_synthesis} has a node truncated using Rule 1.
\end{theorem}
\textbf{Sketch of the Proof} (see the proof in the Appendix). First of all, we prove that the TST returned by Algorithm~\ref{alg:hs_synthesis} is finite for every weakly-complete deterministic TFSM. \\
$\Leftarrow$ We show that if a sequence $(i_1,\delta_1)(i_2,\delta_2)\dots(i_{\ell},\delta_{\ell})$ labels the path from the root to a node truncated using Rule 1, then the timed input sequence $(i_1,\delta_1)(i_2,\delta_1+\delta_2)\dots(i_{\ell},\delta_1+\delta_2+\dots+\delta_{\ell})$ is a homing sequence.\\
$\Rightarrow$ We show that for every shortest HS $\alpha=(i_1,t_1)(i_2,t_2)\dots(i_n,t_n)$ there exists the sequence $(i_1,t'_1)(i_2,t'_2)\dots(i_n,t'_n)$ such that $\alpha'$ is a homing sequence and the sequence $\alpha'=(i_1,t'_1)(i_2,t'_2-t'_1)\dots(i_n,t'_n-t'_{n-1})$ labels the path from the root to the node truncated using Rule 1.

\begin{figure}[ht]
\centering
\begin{minipage}{.45\textwidth}
  \centering
  \begin{tikzpicture}[>=stealth',node distance=4.5cm,semithick,auto, scale=1.0, every node/.style={scale=0.7}]
				\node[state, minimum size=1cm]	(s0)
				    {$s_0$};
                \node[state, right of=s0,minimum size=1cm]	(s1)
				    {$s_1$};
				\node[state, right of=s1,minimum size=1cm]	(s2)
				    {$s_2$};

				{
				\path[->] 
				
                (s0) edge [loop above, align=left] node {
                        $i_1, [2, 3) / o_2, 1$
                        } (s0)
                (s0) edge [midway, above, align=left] node {
                        $i_2, [2, 4) / o_2, 1$
                        } (s1)
                (s1) edge [loop above, align=left] node {
                        $i_1, [2, 3) / o_2, 1$
                        } (s1)
                (s1) edge [bend left, above, align=left] node {
                        $i_2, [2, 4) / o_2, 1$
                        } (s2)
                (s2) edge [midway, below, align=left] node {
                        $i_1, [2, 3) / o_1, 2$
                        } (s1)
                (s2) edge [bend left, below, align=left] node {
                        $i_2, [2, 4) / o_2, 1$
                        } (s0)
					;
				}
\end{tikzpicture}

  \caption{TFSM $\mathcal{S}_2$}
  \label{fig:fsm_A}
\end{minipage}%
\hfill
\begin{minipage}{.45\textwidth}
  \centering
  \begin{tikzpicture}[node distance = 1.7cm, auto,scale=0.68, every node/.style={scale=0.68}]
\tikzstyle{vecArrow} = [thick, decoration={markings,mark=at position
   1 with {\arrow[semithick]{open triangle 60}}},
   double distance=1.4pt, shorten >= 5.5pt,
   preaction = {decorate},
   postaction = {draw,line width=1.4pt, white,shorten >= 4.5pt}]
\tikzstyle{innerWhite} = [semithick, white,line width=1.4pt, shorten >= 4.5pt]
\tikzstyle{decision} =
        [
                diamond,
                draw,
                fill = green!20,
                text width = 6em,
                text badly centered,
                node distance = 2cm,
                inner sep = 0pt
        ]
\tikzstyle{block} =
        [
                rectangle,
                draw,
                text width = 5em,
                text centered,
                rounded corners,
                minimum height = 2em
        ]
\tikzstyle{line} =
        [
                draw,
                -latex'
        ]
\tikzstyle{cloud} =
        [
                draw,
                ellipse,
                fill = red!20,
                node distance = 4cm,
                minimum height = 2em
        ]
\tikzstyle{suite}=[->,>=stealth’,thick,rounded corners=4pt]

        \node [block] (s0) {
            $\overline{s_0, s_1, s_2}$
        };
        
        \node [below of = s0] (i0) {};
        \node [left of = i0] (i0_l) {};
        \node [left of = i0_l] (i0_ll) {};
        
        \node [block, left of = i0_ll] (s1) {
            $\overline{s_1}, \overline{s_0, s_1}$
        };
        \node [block, right of = i0_l] (s2) {
            $\overline{s_0, s_1, s_2}$
        };
        \node [below of = i0] (i1) {};
        \node [left of = i1] (i1_l) {};
        \node [left of = i1_l] (i1_ll) {};
        \node [block, left of = i1_ll] (s3) {
            $\overline{s_0, s_1}$
        };
        \node [block, right of = i1_l] (s4) {
            $\overline{s_1, s_2}$
        };
        \node [below of = i1] (i2) {};
        \node [left of = i2] (i2_l) {};
        \node [left of = i2_l] (i2_ll) {};
        \node [block, accepting, left of = i2_ll] (s5) {
            $\overline{s_1}$
        };
        \node [block, right of = i2_l] (s6) {
            $\overline{s_0, s_2}$
        };
        
        \draw [thick, ->] (s0) -- (s1) node[midway, above=3pt] {$(i_1, 2.5)$};
        \draw [thick, ->] (s0) -- (s2) node[midway, right] {
            $(i_2, 2.5)$
        };
        \draw [thick, ->] (s1) -- (s3) node[midway, right] {$(i_1, 2.25)$};
        \draw [thick, ->] (s1) -- (s4) node[midway, above=6pt] {
            $(i_2, 2.25)$
        };
        \draw [thick, ->] (s4) -- (s5) node[midway, above=6pt] {$(i_1, 2.125)$};
        \draw [thick, ->] (s4) -- (s6) node[midway, right] {
            $(i_2, 2.125)$
        };
        
        
\end{tikzpicture}

  \caption{Fragment of the TST for $\mathcal{S}_2$}
  \label{fig:tst_for_A}
\end{minipage}
\end{figure}

As an example of the truncated successor tree derivation and the application of Algorithm~\ref{alg:hs_synthesis}, we consider TFSM $\mathcal{S}_2$ shown in Fig.~\ref{fig:fsm_A} and the corresponding fragment of the tree shown in Fig.~\ref{fig:tst_for_A}. TFSM $\mathcal{S}_2$ has three states: $s_0$, $s_1$ and $s_2$, therefore, the root is labeled with set $\{s_0, s_1, s_2\}$ denoted as $\overline{s_0, s_1, s_2}$. Since $next\_state_{\mathcal{S}_2}(s_0, (i_1, 2.5)) = s_0$, $next\_state_{\mathcal{S}_2}(s_1, (i_1, 2.5)) = s_1$ and $timed\_out_{\mathcal{S}_2}(s_0, (i_1, 2.5)) = timed\_out_{\mathcal{S}_2}(s_1, (i_1, 2.5)) = \{(o_2, 3.5)\}$, $s_0$ and $s_1$ are found in the same block for $(i_1, 2.5)$-successor. Otherwise, due to the fact that $next\_state_{\mathcal{S}_2}(s_2, (i_1, 2.5)) = s_1$ and $timed\_out_{\mathcal{S}_2}(s_2, (i_1, 2.5)) = \{(o_1, 4.5)\}$, timed input $(i_1, 2.5)$ splits state $s_2$ from states $\{ s_0, s_1 \}$. In another branch, timed input $(i_2, 2.5)$ does not split states from $\overline{s_0, s_1, s_2}$, therefore the corresponding node is truncated using Rule 2. Continuing the same way, we conclude that sequence $(i_1,2.5)(i_2,2.25)(i_1,2.125)$ labels the path from the root, therefore $\alpha = (i_1,2.5)(i_2,4.75)(i_1,6.875)$ is homing for TFSM $\mathcal{S}_2$.

Due to the Definitions~\ref{def:ss} and~\ref{def:hs} of the homing and synchronizing sequences, in order to adapt Algorithm~\ref{alg:hs_synthesis} for the derivation of a shortest SS, we simply need to modify the truncating Rule 1, accordingly. Indeed, in this case, set $P$ should contain only one singleton, which ensures the merging of all initial states to one after the application of the sequence labeling the path to this node. As an example of the derivation of a shortest SS, we again consider TFSM $\mathcal{S}_2$. Repeating the corresponding discussion as with the homing sequence derivation for TFSM $\mathcal{S}_2$, we conclude that $\alpha = (i_1,2.5)(i_2,4.75)(i_1,6.875)$ is also a synchronizing (not only homing) sequence for TFSM $\mathcal{S}_2$, since $\alpha$ brings $\mathcal{S}_2$ from states $s_0, s_1, s_2$ to state $s_1$.
\section{Region FSM \& its properties}\label{sec:fsm_abstractions}
\subsection{Derivation of the region FSM}\label{subsec:1_fsm_abstraction}

To begin with, consider two machines -- TFSM $\mathcal{S}_3$ and FSM $\mathcal{M}_3$ in Figures \ref{fig:example_hs_tfsm} and \ref{fig:example_no_hs_fsm}. It is easy to check that $\mathcal{M}_3$ is constructed from $\mathcal{S}_3$ by ignoring all timed delays and guards and also $\mathcal{M}_3$ has no homing sequence. Although TFSM $\mathcal{S}_3$ has homing sequence $\alpha = (i_1, 1.5)(i_2, 3)$ in the view of Definition \ref{def:hs}. The FSM abstraction above ``forgets'' all timed parameters, i.e., timed guards and output delays of transitions, and the example thus demonstrates that we cannot simply erase all the information about time to compute an HS for a TFSM. To solve the problem, we need to proceed differently. That is the reason why we introduce a \emph{refined} FSM abstraction (\textit{region FSM}) of a TFSM and show that under certain assumptions, we can establish the correspondence between SSs and HSs for the TFSM and its region FSM.

\begin{figure}[ht]
\centering
\begin{minipage}{.6\textwidth}
  \centering
  \begin{tikzpicture}[>=stealth',node distance=3.7cm,semithick,auto, scale=0.8, every node/.style={scale=1}]
    \node[state, minimum size=1cm]  (s0) {$s_0$};
    \node[state, right of=s0,minimum size=1cm]  (s1) {$s_1$};
    \node[state, below of=s0,minimum size=1cm]  (s2) {$s_2$};
    \node[state, right of=s2,minimum size=1cm]  (s3) {$s_3$};

        {
        \path[->] 
        
    (s0) edge [midway, below, align=left] node {
        $i_1,[1,2)/o_1,5$
            } (s1)
    (s0) edge [loop above, above, align=left] node {
       $i_2,[1,2)/o_2,1$
            } (s0)
    (s1) edge [bend right, above, align=left] node {
        $i_1,[1,2)/o_1,1$
            } (s0)
    (s1) edge [loop above, above, align=left] node {
       $i_2,[1,2)/o_2,1$
            } (s1)
    (s2) edge [midway, above, align=left] node[rotate=90] {
        $i_2,[1,2)/o_2,1$
            } (s0)
    (s2) edge [midway, below, align=left] node {
        $i_1,[1,2)/o_1,5$
            } (s3)
    (s3) edge [bend right, above, align=left] node {
        $i_1,[1,2)/o_1,1$
            } (s2)
    (s3) edge [midway, below, align=left] node[rotate=90] {
        $i_2,[1,2)/o_2,1$
            } (s1)
        ;
    }
\end{tikzpicture}

  \caption{TFSM $\mathcal{S}_3$}
  \label{fig:example_hs_tfsm}
\end{minipage}%
\hfill
\begin{minipage}{.38\textwidth}
  \centering
  \begin{tikzpicture}[>=stealth',node distance=3.7cm,semithick,auto, scale=0.8, every node/.style={scale=1}]
    \node[state, minimum size=1cm]  (s0) {$s_0$};
    \node[state, right of=s0,minimum size=1cm]  (s1) {$s_1$};
    \node[state, below of=s0,minimum size=1cm]  (s2) {$s_2$};
    \node[state, right of=s2,minimum size=1cm]  (s3) {$s_3$};

        {
        \path[->] 
        
    (s0) edge [midway, below, align=left] node {
        $i_1/o_1$
            } (s1)
    (s0) edge [loop above, above, align=left] node {
       $i_2/o_2$
            } (s0)
    (s1) edge [bend right, above, align=left] node {
       $i_1/o_1$
            } (s0)
    (s1) edge [loop above, above, align=left] node {
       $i_2/o_2$
            } (s1)
    (s2) edge [midway, above, align=left] node[rotate=90] {
       $i_2/o_2$
            } (s0)
    (s2) edge [midway, below, align=left] node {
       $i_1/o_1$
            } (s3)
    (s3) edge [bend right, above, align=left] node {
       $i_1/o_1$
            } (s2)
    (s3) edge [midway, below, align=left] node[rotate=90] {
       $i_2/o_2$
            } (s1)
        ;
    }
\end{tikzpicture}

  \caption{FSM $\mathcal{M}_3$}
  \label{fig:example_no_hs_fsm}
\end{minipage}
\end{figure}

Given a TFSM $\mathcal{S} = (S, I, O, G, D, h_S)$, $i \in I$, $U_i$ and $V_i$ are minimal left and maximal right boundaries for input $i$ (see Section~\ref{subsec:notation}). Let $G_i = \{ g \in G\ |\ \exists\ s \stackrel{i,g/o,d}{\longrightarrow} s' \in h_S \}$ be the set of all timed intervals guarding the transitions labeled with $i$, $\overline{G_i}$ is obtained by the following procedure. Suppose that $G_i = \{ [u_1, v_1), [u_2, v_2), \dots, [u_n, v_n) \}$, and $p_1, p_2, \dots, p_{x}$, where $x \leq 2n$, are boundary points from $G_i$ arranged in an increasing order. For each consecutive pair of boundary points $p_i$ and $p_{i+1}$, a new interval $[p_i,p_{i+1})$ is derived; thus, $\overline{G_i} = \{ [p_1, p_2), [p_2, p_3), \dots, [p_{x-1}, p_x) \}$. 
Consider, for example, TFSM $\mathcal{S}_1$ (Fig.~\ref{fig:spec_tfsm}) and timed guards for input $i_2$. Set $G_{i_2} = \{ [1, 3),[3,4),[3,5),[5,6),$ $[4,6),[3,6) \}$ contains all possible timed guards for $i_2$. Therefore, boundary points are $1, 3, 4, 5, 6$, thus $\overline{G_{i_2}} = \{ [1, 3),[3,4),[3,5),[5,6) \}$. We introduce $I_G = \{ (i, g)\ |\ i \in I,\ g \in \overline{G_i} \}$ and $O_D = \{ (o, d)\ |\ o \in O,\ d \in D \text{ and } \exists\ s \stackrel{i,g/o,d}{\longrightarrow} s' \in h_S \}$ as the sets of \textit{abstract inputs} and \textit{outputs}, respectively. The \textit{region FSM} of $\mathcal{S}$ is derived as $R(\mathcal{S}) = (S, I_G, O_D, R(h_S))$, where $s \stackrel{(i,g)/(o,d)}{\longrightarrow} s' \in R(h_S)$ if and only if $(i,g) \in I_G$, $(o,d) \in O_D$ and there exists $s \stackrel{i,g'/o,d}{\longrightarrow} s' \in h_S$ such that $g \subseteq g'$. As an example of the region FSM derivation, consider TFSM $\mathcal{S}_4$ (Fig.~\ref{fig:example_not_hs_tfsm}). Since $G_{i_1} = \{ [0,1),[1,2),[0,2) \}$ and $G_{i_2} = \{ [1,3) \}$, we conclude that $\overline{G_{i_1}} = \{ [0,1),[1,2) \}$, $\overline{G_{i_2}} = \{ [1,3) \}$. Thus, $I_G = \{ (i_1, [0, 1)), (i_1, [1, 2)), (i_2, [1, 3)) \}$ and $O_D = \{ (o_1, 1), (o_1, 3), (o_2, 2), (o_2, 4) \}$. Due to the fact that TFSM $\mathcal{S}_4$ has transition $s_0 \stackrel{i_1,[0,2) / o_1, 3}{\rightarrow} s_2$ and $\overline{G_{i_1}} = \{ [0,1),[1,2) \}$, transitions $s_0 \stackrel{(i_1,[0,1)) / (o_1,3)}{\rightarrow} s_2$ and $s_0 \stackrel{(i_1,[1,2)) / (o_1,3)}{\rightarrow} s_2$ are included in $R(h_S)$. The corresponding region FSM $R(\mathcal{S}_4)$ is shown in Fig.~\ref{fig:1_fsm_abstraction}.

Consider the correspondence between timed input/output sequences of $\mathcal{S}$ and their untimed counterparts of $R(\mathcal{S})$\footnote{To introduce the correspondence between a timestamp $t$ and the corresponding timed guard $[t]$, we inherit the Alur\&Dill notation~\cite{alur}.}. Let $(i, t)$ be a timed input, $[t]_{\overline{G_i}} = g$ if there exists $g \in \overline{G_i}$ such that $t \in g$, otherwise $[t]_{\overline{G_i}} = \{ \bot \}$. Given $\alpha = (i_1, t_1)(i_2, t_2) \dots (i_n, t_n)$, $[\alpha]_{I_G}$ defines the \textit{untimed projection} of $\alpha$ over $I_G$ in the following way: $[\alpha]_{I_G} = (i_1, [t_1]_{\overline{G_{i_1}}})(i_2, [t_2-t_1]_{\overline{G_{i_2}}})\dots(i_n, [t_n-t_{n-1}]_{\overline{G_{i_n}}})$. As an example, we again consider TFSM $\mathcal{S}_4$ and $\alpha = (i_1,1)(i_2,3)$. Note that $[1]_{\overline{G_{i_1}}}=[1,2)$ and $[3-1]_{\overline{G_{i_2}}}=[1,3)$, therefore $[\alpha]_{I_G} = (i_1,[1,2))(i_2,[1,3))$. According to the construction of the region FSM, Lemma~\ref{lemma:domain} holds.

\begin{lemma}\label{lemma:domain}
If $\mathcal{S} = (S, I, O, G, D, h_S)$ is a weakly-complete deterministic TFSM with $|h_S|=O(|S|^k)$, then its region FSM $R(\mathcal{S}) = (S, I_G, O_D, R(h_S))$ has the following properties:
\begin{enumerate}
    \item $R(\mathcal{S})$ is deterministic and complete;
    \item A timed input sequence $\alpha \in Dom_{\mathcal{S}}(s)$ if and only if $[\alpha]_{I_G} \in Dom_{R(\mathcal{S})}(s)$ for every $s \in S$, moreover, $next\_state_{\mathcal{S}}(s,\alpha)=next\_state_{R(\mathcal{S})}(s,[\alpha]_{I_G})$;
    \item $|R(h_S)| = O(|S|^{2k})$.
\end{enumerate}
\end{lemma}

Note that unlike the region automaton and FSM abstraction defined for Timed Automaton and Timed FSM in~\cite{alur} and~\cite{bresolin_2014} correspondingly, the size of region FSM remains polynomial with respect to the size of the TFSM. This ensures that the transformation does not introduce exponential growth, making it computationally feasible for practical analysis and verification of TFSM properties. Given a weakly-complete deterministic TFSM $\mathcal{S}$ and its region FSM $R(\mathcal{S})$, the following propositions establish the relations between SSs and HSs for $\mathcal{S}$ and $R(\mathcal{S})$.

\begin{theorem}\label{theorem:ss}
$\alpha$ is an SS for $\mathcal{S}$ if and only if $[\alpha]_{I_G}$ is an SS for $R(\mathcal{S})$.
\end{theorem}
\begin{proof}
$\Rightarrow$ Suppose that $\alpha$ is an SS for $\mathcal{S}$, but $[\alpha]_{I_G}$ is not an SS for $R(\mathcal{S})$. Then there exist $s, s' \in S$ such that $next\_state_{\mathcal{S}}(s, [\alpha]_{I_G}) \neq next\_state_{\mathcal{S}}(s', [\alpha]_{I_G})$. And $next\_state_{\mathcal{S}}(s, \alpha) \neq next\_state_{\mathcal{S}}(s', \alpha)$ (Lemma~\ref{lemma:domain}), it is a contradiction.

$\Leftarrow$ Suppose that $[\alpha]_{I_G}$ is an SS for $R(\mathcal{S})$, but $\alpha$ is not an SS for $\mathcal{S}$. Then there exist $s, s' \in S$ such that $next\_state_{\mathcal{S}}(s, \alpha) \neq next\_state_{\mathcal{S}}(s', \alpha)$. And $next\_state_{\mathcal{S}}(s, [\alpha]_{I_G}) \neq next\_state_{\mathcal{S}}(s', [\alpha]_{I_G})$ (Lemma~\ref{lemma:domain}), it is a contradiction.
\end{proof}

Thus, Theorem~\ref{theorem:ss} provides an algorithm for deriving SSs for Timed FSMs by utilizing their corresponding region FSMs. Theorem~\ref{theorem:hs_for_tfsm} claims that the untimed projection of a homing sequence for TFSM remains homing for the corresponding region FSM.

\begin{theorem}\label{theorem:hs_for_tfsm}
If $\alpha=(i_1,t_1)\dots(i_n,t_n)$ is an HS for TFSM $\mathcal{S}$, then $[\alpha]_{I_G}$ is an HS for $R(\mathcal{S})$.
\end{theorem}
\begin{proof}
First, prove the following claim.
\begin{claim}\label{claim:fsm_outputs}
Given states $s$ and $s'$ of $\mathcal{S}$, if $timed\_out_{\mathcal{S}}(s, \alpha) \neq timed\_out_{\mathcal{S}}(s', \alpha)$, then $out_{R(\mathcal{S})}(s, [\alpha]_{I_G}) \neq out_{R(\mathcal{S})}(s', [\alpha]_{I_G})$.
\end{claim}
\begin{proof}
Note that, $[\alpha]_{I_G} = (i_1, [p_1, q_1))\dots (i_n, [p_n, q_n))$. Assume that \\$out_{R(\mathcal{S})}(s, [\alpha]_{I_G}) = out_{R(\mathcal{S})}(s', [\alpha]_{I_G})$ and $timed\_out_{\mathcal{S}}(s, \alpha) \neq timed\_out_{\mathcal{S}}(s', \alpha)$. Then $[\alpha]_{I_G}$ induces $r_{FSM} = s \xrightarrow[(o_1, d_1)]{(i_1, [p_1, q_1))} \dots \xrightarrow[(o_n, d_n)]{(i_n, [p_n, q_n))} s_n$ at $s$ and $r'_{FSM} = s' \xrightarrow[(o_1, d_1)]{(i_1, [p_1, q_1))} \dots \xrightarrow[(o_n, d_n)]{(i_n, [p_n, q_n))} s'_n$ at $s'$ for $R(\mathcal{S})$.

According to the derivation of $R(\mathcal{S})$, $\alpha$ induces runs $r$ and $r'$ for $\mathcal{S}$: $r = s \xrightarrow[o_1, t_1 + d_1]{i_1, t_1} \dots \xrightarrow[o_n, t_n + d_n]{i_n, t_n} s_n$ at $s$ and $r' = s' \xrightarrow[o_1, t_1 + d_1]{i_1, t_1} \dots \xrightarrow[o_n, t_n + d_n]{i_n, t_n} s'_n$ at $s'$.

Thus, $timed\_out_{\mathcal{S}}(s, \alpha) = timed\_out_{\mathcal{S}}(s', \alpha)$, it is a contradiction.
\end{proof}

Now we prove the Theorem. Assume that $\alpha$ is an HS for $\mathcal{S}$ and $[\alpha]_{I_G}$ is not an HS for $R(\mathcal{S})$. Then there exist $s, s' \in S$ such that $out_{\mathcal{S}}(s, [\alpha]_{I_G}) = out_{\mathcal{S}}(s', [\alpha]_{I_G})$ and $next\_state_{\mathcal{S}}(s, [\alpha]_{I_G}) \neq next\_state_{\mathcal{S}}(s', [\alpha]_{I_G})$. Then due to the derivation of $R(\mathcal{S})$ it holds that $next\_state_{\mathcal{S}}(s, \alpha) \neq next\_state_{\mathcal{S}}(s', \alpha)$. Since $\alpha$ is an HS for $\mathcal{S}$, we conclude that $timed\_out_{\mathcal{S}}(s, \alpha) \neq timed\_out_{\mathcal{S}}(s', \alpha)$.
Thus, $timed\_out_{\mathcal{S}}(s, \alpha) \neq timed\_out_{\mathcal{S}}(s', \alpha)$ and at the same time it holds that $out_{R(\mathcal{S})}(s, [\alpha]_{I_G}) = out_{R(\mathcal{S})}(s', [\alpha]_{I_G})$, it is a contradiction (Claim~\ref{claim:fsm_outputs}).
\end{proof}

\begin{figure}[ht]
\centering
\begin{minipage}{.45\textwidth}
  \centering
  \begin{tikzpicture}[>=stealth',node distance=4.3cm,semithick,auto, scale=0.8, every node/.style={scale=1}]
	\node[state, minimum size=1cm]	(s0) {$s_0$};
    \node[state, right of=s0,minimum size=1cm]	(s2) {$s_2$};
    \node[state, below of=s0,minimum size=1cm]	(s1) {$s_1$};
    \node[state, right of=s1,minimum size=1cm]	(s3) {$s_3$};

				{
				\path[->] 
				
    (s0) edge [midway, above, align=left] node {
        $i_1,[0,2)/o_1,3$
            } (s2)
    (s0) edge [loop above, align=left] node {
       $i_2,[1,3)/o_1,1$
            } (s0)
    (s1) edge [bend left, above, align=left] node {
        $i_1,[0,2)/o_1,3$
            } (s3)
    (s1) edge [midway, above, align=left] node[rotate=90] {
        $i_2,[1,3)/o_1,1$
            } (s0)
    (s2) edge [bend left, below, align=left] node {
        $i_1,[0,1)/o_1,3$ \\
        $i_1,[1,2)/o_2,4$
            } (s0)
    (s2) edge [midway, below, align=left] node[rotate=90] {
        $i_2,[1,3)/o_2,2$
            } (s3)
    (s3) edge [midway, below, align=left] node {
        $i_1,[0,2)/o_2,4$
            } (s1)
    (s3) edge [loop below, align=left] node {
       $i_2,[1,3)/o_2,2$
            } (s3)
        ;
		}
\end{tikzpicture}
  \caption{TFSM $\mathcal{S}_4$}
  \label{fig:example_not_hs_tfsm}
\end{minipage}%
\hfill
\begin{minipage}{.45\textwidth}
  \centering
  \begin{tikzpicture}[>=stealth',node distance=4.3cm,semithick,auto, scale=0.8, every node/.style={scale=1}]
	\node[state, minimum size=1cm]	(s0) {$s_0$};
    \node[state, right of=s0,minimum size=1cm]	(s2) {$s_2$};
    \node[state, below of=s0,minimum size=1cm]	(s1) {$s_1$};
    \node[state, right of=s1,minimum size=1cm]	(s3) {$s_3$};

				{
				\path[->] 
				
    (s0) edge [midway, above, align=left] node {
        $(i_1,[0,1))/(o_1,3)$ \\
        $(i_1,[1,2))/(o_1,3)$
            } (s2)
    (s0) edge [loop above, align=left] node {
       $(i_2,[1,3))/(o_1,1)$
            } (s0)
    (s1) edge [bend left, above, align=left] node {
        $(i_1,[0,1))/(o_1,3)$ \\
        $(i_1,[1,2))/(o_1,3)$
            } (s3)
    (s1) edge [midway, above, align=left] node[rotate=90] {
        $(i_2,[1,3))/(o_1,1)$
            } (s0)
    (s2) edge [bend left, below, align=left] node {
        $(i_1,[0,1))/(o_1,3)$ \\
        $(i_1,[1,2))/(o_2,4)$
            } (s0)
    (s2) edge [midway, below, align=left] node[rotate=90] {
        $(i_2,[1,3))/(o_2,2)$
            } (s3)
    (s3) edge [midway, below, align=left] node {
        $(i_1,[0,1))/(o_2,4)$ \\
        $(i_1,[1,2))/(o_2,4)$
            } (s1)
    (s3) edge [loop below, align=left] node {
        $(i_2,[1,3))/(o_2,2)$
            } (s3)
        ;
		}
\end{tikzpicture}

  \caption{Region FSM $R(\mathcal{S}_4)$}
  \label{fig:1_fsm_abstraction}
\end{minipage}
\end{figure}

A related question arises: Does the converse hold ? Specifically, does any timed input sequence such that its projection is a homing sequence for $R(\mathcal{S})$ remain homing for $\mathcal{S}$ ? The following example illustrates that this is not always the case. Consider TFSM $\mathcal{S}_4$ (Fig.~\ref{fig:example_not_hs_tfsm}), its region FSM $R(\mathcal{S}_4)$ (Fig.~\ref{fig:1_fsm_abstraction}) and $\alpha = (i_1,1)(i_2,3)$. Since $timed\_out_{\mathcal{S}_4}(s_0,\alpha) = timed\_out_{\mathcal{S}_4}(s_3,\alpha) = \{(o_1,4)(o_2,5)\}$, but $next\_state_{\mathcal{S}_4}(s_0,\alpha) = s_3$ and $next\_state_{\mathcal{S}_4}(s_3,\alpha) = s_0$, we conclude that $\alpha$ is not a homing sequence for $\mathcal{S}_4$. Now consider untimed projection $[\alpha]_{I_G} = (i_1,[1,2))(i_2,[1,3))$ of $\alpha$. Since $out_{R(\mathcal{S}_4)}(s_0, [\alpha]_{I_G}) = out_{R(\mathcal{S}_4)}(s_1, [\alpha]_{I_G})$ and $next\_state_{R(\mathcal{S}_4)}(s_0, [\alpha]_{I_G}) = next\_state_{R(\mathcal{S}_4)}(s_1, [\alpha]_{I_G})$, while $out_{R(\mathcal{S}_4)}(s_2, [\alpha]_{I_G}) = out_{R(\mathcal{S}_4)}(s_3, [\alpha]_{I_G})$ and $next\_state_{R(\mathcal{S}_4)}(s_2, [\alpha]_{I_G}) = next\_state_{R(\mathcal{S}_4)}(s_3, [\alpha]_{I_G})$, we conclude that $[\alpha]_{I_G}$ is an HS for $R(\mathcal{S}_4)$.

We first discuss why a timed input sequence might not be homing for a TFSM while its untimed projection is a homing sequence for the region FSM. The primary reason is the permutation of outputs, which can prevent the sequence $\alpha$ from splitting two different states (similar to Lemma~\ref{lemma:prolongation_not_hs}). Otherwise, Lemma~\ref{lemma:tfsm_outputs} claims that it is not the case for non-integer timed input sequences.

\begin{lemma}\label{lemma:tfsm_outputs}
Given a TFSM $\mathcal{S}$, its projection $R(\mathcal{S})$, states $s$ and $s'$, and a non-integer timed input sequence $\alpha$, the following holds: if $out_{R(\mathcal{S})}(s, [\alpha]_{I_G}) \neq out_{R(\mathcal{S})}(s', [\alpha]_{I_G})$, then $timed\_out_{\mathcal{S}}(s, \alpha) \neq timed\_out_{\mathcal{S}}(s', \alpha)$.
\end{lemma}
\begin{proof}
$\alpha=(i_1,t_1)\dots(i_n,t_n)$ induces the following runs for $\mathcal{S}$: $r = s \xrightarrow[o_1, t_1 + d_1]{i_1, t_1} \dots \xrightarrow[o_n, t_n + d_n]{i_n, t_n} s_n$ at $s$ and $r' = s' \xrightarrow[o'_1, t_1 + d'_1]{i_1, t_1} \dots \xrightarrow[o'_n, t_n + d'_n]{i_n, t_n} s'_n$ at $s'$.

Given $r\downarrow_{O} = (o_1, t_1 + d_1)\dots(o_n, t_n + d_n)$, $timed\_out_{\mathcal{S}}(s, \alpha) = \{ (o_{j_1}, t_{j_1} + d_{j_1})\dots(o_{j_n}, t_{j_n} + d_{j_n})\}$ is such a permutation $j$ of $r\downarrow_{O}$ that $t_{j_1} + d_{j_1}\leq \dots \leq t_{j_n} + d_{j_n}$. Similarly, given $r'\downarrow_{O} = (o'_1, t_1 + d'_1)\dots(o'_n, t_n + d'_n)$, $timed\_out_{\mathcal{S}}(s, \alpha)$ $ = \{(o'_{k_1}, t_{k_1} + d'_{k_1})\dots(o'_{k_n}, t_{k_n} + d'_{k_n})\}$ is such a permutation $k$ of $r'\downarrow_{O}$ that $t_{k_1} + d_{k_1}\leq \dots \leq t_{k_n} + d_{k_n}$. Assume that $out_{R(\mathcal{S})}(s, [\alpha]_{I_G}) $ $\neq out_{R(\mathcal{S})}(s', [\alpha]_{I_G})$ and $timed\_out_{\mathcal{S}}(s, \alpha) = timed\_out_{\mathcal{S}}(s', \alpha)$, then we conclude that $j$ and $k$ are different permutations, and $t_{j_1} + d_{j_1} = t_{k_1} + d'_{k_1}$, $t_{j_2} + d_{j_2} = t_{k_2} + d'_{k_2}$ \dots, $t_{j_n} + d_{j_n} = t_{k_n} + d'_{k_n}$. Therefore, $|t_{j_1} - t_{k_1}| = |d'_{k_1} - d_{j_1}| \in \mathbb{N}^+_0 $, $|t_{j_2} - t_{k_2}| = |d'_{k_2} - d_{j_2}| \in \mathbb{N}^+_0$ \dots, $|t_{j_n} - t_{k_n}| = |d'_{k_n} - d_{j_n}| \in \mathbb{N}^+_0$, it is a contradiction with $\alpha$ being non-integer timed input sequence.
\end{proof}

Lemma~\ref{lemma:tfsm_outputs} establishes that if every untimed input sequence of the region FSM corresponds to at least one non-integer timed input sequence in the original TFSM, then the correspondence between their homing sequences can be set up. This result leads to Theorem~\ref{theorem:hs_bijection}.

\begin{theorem}\label{theorem:hs_bijection}
Given a TFSM $\mathcal{S}$, non-integer timed input sequence $\alpha$ is an HS for $\mathcal{S}$ if and only if $[\alpha]_{I_G}$ is an HS for $R(\mathcal{S})$.
\end{theorem}
\begin{proof}
$\Rightarrow$ Let $[\alpha]_{I_G}$ be an HS of $R(\mathcal{S})$ and $\alpha$ be a non-integer timed sequence corresponding to $[\alpha]_{I_G}$, assume that $\alpha$ is not an HS for $\mathcal{S}$, then there exist states $s, s'$ of $\mathcal{S}$ such that $timed\_out_{\mathcal{S}}(s, \alpha) = timed\_out_{\mathcal{S}}(s', \alpha)$ and $next\_state_{\mathcal{S}}(s, \alpha) \neq next\_state_{\mathcal{S}}(s', \alpha)$. Due to the derivation of $R(\mathcal{S})$ it holds that $next\_state_{R(\mathcal{S})}(s, [\alpha]_{I_G}) \neq next\_state_{R(\mathcal{S})}(s', [\alpha]_{I_G})$. Since $[\alpha]_{I_G}$ is a homing sequence, we conclude that $out_{R(\mathcal{S})}(s, [\alpha]_{I_G}) \neq out_{R(\mathcal{S})}(s', [\alpha]_{I_G})$. Thus, we conclude that $out_{R(\mathcal{S})}(s, [\alpha]_{I_G}) \neq out_{R(\mathcal{S})}(s', [\alpha]_{I_G})$, $timed\_out_{\mathcal{S}}(s, \alpha) $ $= timed\_out_{\mathcal{S}}(s', \alpha)$, it is a contradiction.

$\Leftarrow$ See Theorem~\ref{theorem:hs_for_tfsm}.
\end{proof}

The definition of left-closed and right-open intervals implies that the intersection of all pairs of timed guards $g, g' \in G$ is either empty or is a left-closed and right-open interval. This leads us to the following Corollary of Theorem~\ref{theorem:hs_bijection}.
\begin{corollary}\label{corollary:hs_criterium}
A TFSM $\mathcal{S}$ has an HS if and only if $R(\mathcal{S})$ has an HS.
\end{corollary}

Let $\mathcal{S}$ be a TFSM with $n$ states. Lemma~\ref{lemma:domain} claims that $R(\mathcal{S})$ has the polynomial size with respect to the size of $\mathcal{S}$, while Theorem~\ref{theorem:ss} and Corollary~\ref{corollary:hs_criterium} establish the correspondence between SSs/HSs for $R(\mathcal{S})$ and $\mathcal{S}$. The latter allows to draw conclusions about the complexity of the SS/HS existence check for TFSMs with output delays. Namely, checking if $\mathcal{S}$ has an SS/HS can be done in polynomial time with respect to $n$. At the same time, the problem of deriving a shortest SS/HS is NP-hard. Theorem~\ref{theorem:hs_bijection} also gives the upper bound on the polynomial length of an HS, when the $R(\mathcal{S})$ is reduced and connected, namely $O(n^2)$. Naturally, a question arises: Is it possible to establish a similar correspondence not only for TFSMs with left-closed and right-open (left-open and right-closed) intervals ? The next section aims to answer this question.

\subsection{Properties of the region FSM for a TFSM with point intervals} \label{subsec:1_fsm_abstraction_for_point_intervals}

\begin{figure}[ht]
\centering
\begin{minipage}{.33\textwidth}
  \centering
  \begin{tikzpicture}[>=stealth',node distance=2.5cm,semithick,auto, scale=0.8, every node/.style={scale=1}]
    \node[state, minimum size=1cm]  (s0) {$s_0$};
    \node[state, right of=s0,minimum size=1cm]  (s1) {$s_1$};
    \node[state, below of=s0,minimum size=1cm]  (s2) {$s_3$};
    \node[state, right of=s2,minimum size=1cm]  (s3) {$s_2$};

        {
        \path[->] 
        
    (s0) edge [bend left, above, align=left] node {
        $i_1,[1,1]/o_1,2$
            } (s1)
    (s1) edge [bend left, below, align=left] node[rotate=90] {
        $i_1,[1,1]/o_1,2$
            } (s3)
    (s2) edge [bend left, above, align=left] node[rotate=90] {
        $i_1,[1,1]/o_1,1$
            } (s0)
    (s3) edge [bend left, below, align=left] node {
        $i_1,[1,1]/o_1,3$
            } (s2)
        ;
    }
\end{tikzpicture}

  \caption{TFSM $\mathcal{B}_4$}
  \label{fig:counterexample_tfsm}
\end{minipage}
\hfill
\begin{minipage}{.33\textwidth}
  \centering
  \begin{tikzpicture}[>=stealth',node distance=2.5cm,semithick,auto, scale=0.8, every node/.style={scale=1}]
    \node[state, minimum size=1cm]  (s0) {$s_0$};
    \node[state, right of=s0,minimum size=1cm]  (s1) {$s_1$};
    \node[state, below of=s0,minimum size=1cm]  (s2) {$s_3$};
    \node[state, right of=s2,minimum size=1cm]  (s3) {$s_2$};

        {
        \path[->] 
        
    (s0) edge [bend left, above, align=left] node {
        $(i_1,[1,1])/(o_1,2)$
            } (s1)
    (s1) edge [bend left, below, align=left] node[rotate=90] {
        $(i_1,[1,1])/(o_1,2)$
            } (s3)
    (s2) edge [bend left, above, align=left] node[rotate=90] {
        $(i_1,[1,1])/(o_1,1)$
            } (s0)
    (s3) edge [bend left, below, align=left] node {
        $(i_1,[1,1])/(o_1,3)$
            } (s2)
        ;
    }
\end{tikzpicture}
  \caption{Region FSM $R(\mathcal{B}_4)$}
  \label{fig:counterexample_fsm}
\end{minipage}%
\hfill
\begin{minipage}{.25\textwidth}
  \centering
  \begin{tikzpicture}[node distance = 1.5cm, auto,scale=0.68, every node/.style={scale=1.0}]
\tikzstyle{vecArrow} = [thick, decoration={markings,mark=at position
   1 with {\arrow[semithick]{open triangle 60}}},
   double distance=1.4pt, shorten >= 5.5pt,
   preaction = {decorate},
   postaction = {draw,line width=1.4pt, white,shorten >= 4.5pt}]
\tikzstyle{innerWhite} = [semithick, white,line width=1.4pt, shorten >= 4.5pt]
\tikzstyle{decision} =
        [
                diamond,
                draw,
                fill = green!20,
                text width = 6em,
                text badly centered,
                node distance = 2cm,
                inner sep = 0pt
        ]
\tikzstyle{block} =
        [
                rectangle,
                draw,
                text width = 5em,
                text centered,
                rounded corners,
                minimum height = 2em
        ]
\tikzstyle{line} =
        [
                draw,
                -latex'
        ]
\tikzstyle{cloud} =
        [
                draw,
                ellipse,
                fill = red!20,
                node distance = 4cm,
                minimum height = 2em
        ]
\tikzstyle{suite}=[->,>=stealth’,thick,rounded corners=4pt]

        \node [block] (s0) {
            $\overline{s_0, s_1, s_2, s_3}$
        };
        \node [block, below of = s0] (s1) {
            $\overline{s_1, s_2}, \overline{s_3}, \overline{s_0}$
        };
        \node [block, accepting, below of = s1] (s2) {
            $\overline{s_2}, \overline{s_3}, \overline{s_0}, \overline{s_1}$
        };
        
        \draw [thick, ->] (s0) -- (s1) node[midway, right] {$(i_1, 1)$};
        \draw [thick, ->] (s1) -- (s2) node[midway, right] {
            $(i_1, 1)$
        };
\end{tikzpicture}

  \caption{TST for $\mathcal{B}_4$}
  \label{fig:tst_for_B4}
\end{minipage}
\end{figure}

We previously focused on checking the existence and deriving HSs for a certain class of TFSMs. We have proven that a left-closed and right-open TFSM has an SS (HS) if and only if its region FSM has an SS (HS) (Theorem~\ref{theorem:ss} and Corollary~\ref{corollary:hs_criterium}). In this section, we consider another class of TFSMs. We say that $\mathcal{S}_p$ is a \emph{TFSM with only point intervals} if every timed guard $g$ of $\mathcal{S}_p$ is a point interval, i.e., $g = [u, u]$ for $u \in \mathbb{N}^{+}$.
\begin{lemma}\label{lemma:point_intervals_properties}
If $\mathcal{S}_p$ is a TFSM with only point intervals, then for every $s \in S$ and for every $\alpha \in Dom_{\mathcal{S}_p}(s)$ it holds that:
\begin{enumerate}
    \item $\alpha$ is an integer timed input sequence;
    \item $next\_state_{\mathcal{S}_p}(s, \alpha)$ returns a singleton.
\end{enumerate}
\end{lemma}
\begin{proof}
1. Given $s \in S$ and $\alpha = (i_1,t_1)\dots(i_n,t_n) \in Dom_{\mathcal{S}_p}(s)$. Since $\alpha \in Dom_{\mathcal{S}_p}(s)$, it holds that $t_j - t_{j-1} \in g$ for $j \in \{ 1, \dots, n \}$ and for some $g \in G$. Due to the fact that $g = [u, u]$ for $u \in \mathbb{N}^{+}$, we conclude that $\alpha$ is an integer timed input sequence.

2. Since $\mathcal{S}_p$ is deterministic, for every $s \in S$ and for every $\alpha \in Dom_{\mathcal{S}_p}(s)$ it holds that $next\_state_{\mathcal{S}_p}(s, \alpha)$ is a singleton.
\end{proof}

Lemma~\ref{lemma:point_intervals_properties} allows to conclude that Theorem~\ref{theorem:ss} is also valid for TFSMs with only point intervals. Namely, $\mathcal{S}_p$ has an SS if and only if $R(\mathcal{S}_p)$ has an SS. However, it is not the case for HSs. Consider TFSM $\mathcal{B}_4$ shown in Fig.~\ref{fig:counterexample_tfsm}, its region FSM $R(\mathcal{B}_4)$ shown in Fig.~\ref{fig:counterexample_fsm} and $\alpha_{fsm}=(i_1,[1,1])(i_1,[1,1])$. Since $out_{R(\mathcal{B}_4)}(s_0,\alpha_{fsm}) = \{(o_1,2)(o_1,2)\}$, $out_{R(\mathcal{B}_4)}(s_1,\alpha_{fsm}) = \{(o_1,2)(o_1,3)\}$,\\ $out_{R(\mathcal{B}_4)}(s_2,\alpha_{fsm})= \{(o_1,3)(o_1,1)\}$ and $out_{R(\mathcal{B}_4)}(s_3,\alpha_{fsm}) = \{(o_1,1)(o_1,2)\}$, we conclude that $\alpha_{fsm}$ is an HS for $R(\mathcal{B}_4)$. At the same time, consider the behavior of TFSM $\mathcal{B}_4$ on $\alpha=(i_1,1)(i_1,2)$, note that $\alpha_{fsm}=[\alpha]_{I_G}$. Due to the fact that $timed\_out_{\mathcal{B}_4}(s_0,\alpha)=timed\_out_{\mathcal{B}_4}(s_2,\alpha)=\{(o_1,3)(o_1,4)\}$, $next\_state_{\mathcal{B}_4}(s_0,\alpha) $ $ \neq next\_state_{\mathcal{B}_4}(s_2,\alpha)$, $\alpha$ is not an HS for $\mathcal{B}_4$. Moreover, there exists a class of TFSMs that cannot be homed while their region FSMs can be homed (Theorem~\ref{theorem:B_n^t_is_not_homed}). We define TFSM $\mathcal{B}_n = (S_n, I, O, G, D, h_{S_n})$ in the following way: 
\begin{itemize}
    \item $S_n = \{s_0, \ldots, s_{n-1}\}$;
    \item $I = \{i_1\}$ and $O = \{o_1\}$;
    \item $G=\{[1,1]\}$ and $D = \{ 1, 2, 3 \}$;
    \item $h_{S_n} = \{(s_i,i_1,[1,1],o_1,2,s_{i+1}) : 0 \leq i \leq n-3\} \cup \\ \cup \{ (s_{n-2},i_1,[1,1],o_1,3,s_{n-1}),  (s_{n-1},i_1,[1,1],o_1,1,s_0)\}$.
\end{itemize}
In other words, input $i_1$ acts like a cyclic permutation on the set of states, and all the outputs, except at states $s_{n-2}$ and $s_{n-1}$, are $(o_1,2)$. It is easy to see that the machine in the Fig. \ref{fig:counterexample_fsm} is a machine $\mathcal{B}_n$ for $n = 4$.

\begin{theorem}\label{theorem:B_n^t_is_not_homed}
For any $n > 3$ we have that $R(\mathcal{B}_n)$ can be homed, but $\mathcal{B}_n$ cannot be homed.
\end{theorem}
\begin{proof}
First, observe that $out_{R(\mathcal{B}_n)}(\cdot, \cdot)$ and $timed\_out_{\mathcal{B}_n}(\cdot, \cdot)$ always give us a singleton, so we can omit the brackets. Let $a = (i_1,[1,1])$. Observe that a sequence $a^{n-2}$ is an HS for $R(\mathcal{B}_n)$.  Indeed, $out_{R(\mathcal{B}_n)}(s_0, a^{n-2}) = (o_1,2)^{n-2}$ and for any $0 < i < n$ it holds, that $out_{R(\mathcal{B}_n)}(s_i, a^{n-2})$ has $(o_1, 1)$ in the $(n-i+1)$-th position, whilst for $j \neq i$ $(n-i+1)$-th position of $out_{R(\mathcal{B}_n)}(s_j, a^{n-2})$ is occupied by either $(o_1,2)$ or $(o_1,1)$. Thus, $a^{n-2}$ is a homing sequence for $R(\mathcal{B}_n)$.

All sequences enabled for $\mathcal{B}_n$ are of the form $\alpha = (i_1,1)(i_1,2) \ldots (i_1,k)$, for $k \in \mathbb{N}$. First, we will show that for any sequence $\alpha_l = (i_1,1) \ldots (i_1,l)$ such that $l \leq n$, at least one pair of states from set $\{s_0, s_1, s_2\}$ produces identical timed output. Obviously, as long as $l < n-3$, it holds $timed\_out_{B_n}(s_0, \alpha_l) = timed\_out_{B_n}(s_1, \alpha_l) = timed\_out_{B_n}(s_2,\alpha_l) =  (o_1, 2+1)(o_1, 2+2) \ldots (o_1, 2+l) $. Consider four cases:

\noindent
\textbf{Case 1:} $l = n-3$, $timed\_out_{B_n}(s_0, \alpha_l) = timed\_out_{B_n}(s_1, \alpha_l) = (o_1, 3)(o_1, 4)$ $ \ldots (o_1, n-1) $ and $timed\_out_{B_n}(s_2, \alpha_l) = (o_1, 3)(o_1, 4) \ldots (o_1, n-2)(o_1, n)$.

\noindent
\textbf{Case 2:} $l = n-2$, $timed\_out_{B_n}(s_0, \alpha_l) = timed\_out_{B_n}(s_2, \alpha_l) = (o_1, 3)(o_1, 4)$ $ \ldots (o_1, n)$ and $timed\_out_{B_n}(s_2, \alpha_l) = (o_1, 3)(o_1, 4) \ldots (o_1, n-1)(o_1, n+1)$.

\noindent
\textbf{Case 3:} $l = n-1$, $timed\_out_{B_n}(s_1, \alpha_l) = timed\_out_{B_n}(s_2, \alpha_l) = (o_1, 3)(o_1, 4)$ $ \ldots (o_1, n+1)$ and $timed\_out_{B_n}(s_2, \alpha_l) = (o_1, 3)(o_1, 4) \ldots (o_1, n)(o_1, n+2)$.

\noindent
\textbf{Case 4:} $l = n$, $timed\_out_{B_n}(s_0, \alpha_l) = timed\_out_{B_n}(s_1, \alpha_l) =$\\ $= timed\_out_{B_n}(s_2, \alpha_l) = (o_1, 3)(o_1, 4) \ldots(o_1, n + 2)$.

Since input $i_1$ induces a cyclic permutation, and $timed\_out_{B_n}(s_0, \alpha_n)  = \\ timed\_out_{B_n}(s_1, \alpha_n) = timed\_out_{B_n}(s_2, \alpha_n)$, then our argument can be extended for any $l > n$.
\end{proof}

In Section~\ref{subsec:notation} we presented Algorithm~\ref{alg:hs_synthesis}, which returns a shortest HS for TFSMs with left-closed and right-open intervals. To evaluate whether the same algorithm can be applied to TFSMs with only point intervals\footnote{For such a TFSM, the edges of the tree are labeled with $(i,u)$ for $[u,u]\in G$, and not with $(i,u+\theta)$.}, consider the tree shown in Fig.~\ref{fig:tst_for_B4} for TFSM $\mathcal{B}_4$ which is truncated using Rule 1; $\alpha = (i_1, 1)(i_1, 2)$ is not homing for $\mathcal{B}_4$ (as discussed earlier), however, $(i_1, 1)(i_1, 1)$ labels the path from the root to the terminal node (Rule 1). The reason is that a TFSM with only point intervals does not have any non-integer timed input sequence in the domain (Lemma~\ref{lemma:point_intervals_properties}). In particular, the permutation of outputs for an integer timed input sequence $\alpha_{int}$ can lead to the following: even if $\alpha_{int}$ is homing for a state pair at $j$-th level, the prolongation of $\alpha_{int}$ might stop being homing for the same pair of states at $\ell$-th level for $j < \ell$ (Lemma~\ref{lemma:prolongation_not_hs}). Therefore, one of the ways to modify the truncated tree derivation is to take into account also timed outputs that can be produced after the execution of a timed input sequence (\emph{timed output tail}), while deriving the successor tree. Namely, instead of defining the successor function solely over states, we propose defining it over states and timed output tails. This refinement ensures a more precise unrolling of the TFSM’s behavior, allowing for the correct derivation of HSs in the presence of point intervals. Let $\mathcal{S}_p$ be a TFSM with only point intervals, $s \in S$ and $\alpha = (i_1, t_1)\dots(i_n,t_n) \in Dom_{\mathcal{S}_p}(s)$ while $t(\alpha) = t_n$ denotes the execution time of $\alpha$.
We define the \textit{timed output tail} (or \textit{tail}, for short) of the response of $\mathcal{S}_p$ to $\alpha$ at $s$ as $timed\_out_{\mathcal{S}_p \geq t(\alpha)}(s, \alpha) = \{(o_1,\tau_1)^{k_1},\dots,(o_m,\tau_m)^{k_m} \}$ which for $k_j > 0$ is the set of all (possibly repeated) timed outputs $(o_j,\tau_j)$, that are in $timed\_out_{\mathcal{S}_p}(s, \alpha)$ and are produced at or after $t(\alpha)$. Formally, $(o, \tau-\tau(\alpha))^{k} \in timed\_out_{\mathcal{S}_p \geq t(\alpha)}(s, \alpha)$ if and only if timed output $(o, \tau)$ occurs $k$ times in every timed output sequence of $timed\_out_{\mathcal{S}_p}(s, \alpha)$ and $\tau \geq t(\alpha)$, $|timed\_out_{\mathcal{S}_p \geq t(\alpha)}(s, \alpha)| = k_1+\dots+k_m$. As an example, consider $\gamma=(i_1,1)(i_1,2)(i_1,3)(i_1,4)$ applied at state $s_0$ of TFSM $\mathcal{B}_4$. Since $t(\gamma) = 4$ and $timed\_out_{\mathcal{B}_4}(s_0, \gamma) = \{ (o_1,3)(o_1,4)(o_1,5)(o_1,6) \}$, the tail of $\gamma$ at $s_0$ is $\{ (o_1,0),(o_1,1),(o_1,2) \}$. Let $\mathcal{T}_{out}$ be the set of timed output tails of $\mathcal{S}_p$, Lemma~\ref{lemma:t_out_len} establishes the upper bounds on the cardinalities of the reachable tails and $\mathcal{T}_{out}$.

\begin{lemma}\label{lemma:t_out_len}
Let $\mathcal{S}_p = (S, I, O, G, D, h_S)$ be a (possibly partial) TFSM with only point intervals and $|h_S| = O(|S^k|)$, the following holds:
\begin{enumerate}
    \item $|timed\_out_{\mathcal{S}_p \geq t(\alpha)}(s, \alpha)| \leq \outlen$ for every $s \in S$ and $\alpha \in Dom_{\mathcal{S}_p}(s)$;
    \item $|\mathcal{T}_{out}| = O(|S|^{3k\cdot \outlen})$.
\end{enumerate}
\end{lemma}
\textbf{Sketch of the Proof} (see the proof in the Appendix). To prove Point 1, we establish that the maximal possible number of inputs applied between $t(\alpha) - \max\{D\}$ and $t(\alpha)$ is exactly $\outlen$. Since exactly one output is produced for every input, the claim holds.
To prove Point 2, we construct multiset $T$ of all possible timed output tails for $\mathcal{S}_p$ and show that $\mathcal{T}_{out} = \{cut\_right(\kappa,0) : \kappa \in \mathcal{T}\} \cup \{\epsilon\}$. In the second step, we show that $|\mathcal{T}_{out}| \leq \sum_{i=1}^{\outlen}(|O||D||G|)^i = \tlen$. Given that $|h_S| = O(|S|^k)$, the claim holds.

Let $tail = \{(o_1,\tau_1)^{k_1},\dots,(o_m,\tau_m)^{k_m} \}$ be a timed output tail of $\mathcal{T}_{out}$, we define the following operations over $tail$:
i) $cut\_left(tail,t) = \{ (o, \tau)^k\in tail\ |\ \tau < t \}$ is the set of the timed outputs of $tail$ such that all their timestamps are less than $t$, ii) $cut\_right(tail,t) = \{ (o, \tau)^k\in tail\ |\ \tau \geq t \}$ is the set of the timed outputs of $tail$ such that all their timestamps are greater than or equal to $t$, and iii) $shift(tail,t) = \{ (o, \tau + t)^k\ |\ (o, \tau)^k \in tail \}$. We say that $(s',tail')=(i,t)/ \{\dots,(o_j,\tau_j)^{k_j},\dots \}$-$succ((s, tail))$ if $s' = next\_state_{\mathcal{S}_p}(s, (i, t))$, $\{(o, t+d)\} = timed\_out_{\mathcal{S}_p}(s, (i, t)) \}$, $\{\dots,(o_j,\tau_j)^{k_j},\dots\}$ $=cut\_left(shift(tail,-t),0)$ and $tail'=cut\_right(shift(tail,-t),0) \cup \{ (o,d) \}$. As an example, consider TFSM $\mathcal{B}_4$ (Fig.~\ref{fig:counterexample_tfsm}) at state $s_2$ when timed outputs $(o_1, 0)$ and $(o_1, 1)$ are pending. Since $s_2 \stackrel{i_1, [1,1] / o_1, 3}{\rightarrow} s_3$, \\$cut\_right(shift(\{(o_1,0),(o_1,1)\},-1),0)\cup \{ (o_1, 3)\}=\{(o_1,0), (o_1,3)\}$ and $cut\_left(shift($ $\{(o_1,0),(o_1,1)\},-1),0)=\{(o_1,-1)\}$, it holds that \\$(s_3, \{(o_1,0),(o_1,3)\})=(i_1, 1)/\{(o_1,-1)\}$-$succ(s_2, \{(o_1,0),(o_1,1)\})$. Function $(i,t)/\{\dots,(o_j,\tau_j)^{k_j},\dots\}$-$succ: S \times \mathcal{T}_{out} \rightarrow S \times \mathcal{T}_{out}$ can be extended to operate over the subsets of $S \times \mathcal{T}_{out}$. In particular, let $Q_1, Q_2$ be subsets of $S\times \mathcal{T}_{out}$, we say that $Q_2 = (i,t)/\{\dots,(o_j,\tau_j)^{k_j},\dots\}$-$succ(Q_1)$ if and only if for every $q' \in Q_2$ there exists $q \in Q_1$ such that $q' = (i,t)/\{\dots,(o_j,\tau_j)^{k_j},\dots\}$-$succ(q)$.

Given a TFSM $\mathcal{S}_p$ with only point intervals. In order to derive an HS for $\mathcal{S}_p$, we modify the successor function in Algorithm~\ref{alg:hs_synthesis} as discussed above together with the labeling of nodes. Instead of labeling them with the set of subsets of $S$, we label them with the set of subsets of $S \times \mathcal{T}_{out}$, without changing truncated rules.
Therefore, for TFSM $\mathcal{B}_4$, the root of the tree will be labeled with $\overline{(s_0,\{ \varepsilon \}),(s_1,\{ \varepsilon \}),(s_2,\{ \varepsilon \}),(s_3,\{ \varepsilon \})}$; and moreover it will only have one branch which is truncated using Rule 2 (Fig.~\ref{fig:tst_configuration}). Therefore, $\mathcal{B}_4$ does not have a homing sequence.
The following theorem establishes an upper bound on the length of a shortest HS for TFSMs with only point intervals (when it exists).

\begin{theorem}\label{theorem:point_tfsm_hs_bound}
Let $\mathcal{S}_p = (S, I, O, G, D, h_S)$ be a (possibly partial) TFSM with only point intervals, $|h_S| = O(|S|^k)$ and $\alpha$ be a shortest HS for $\mathcal{S}_p$, it holds that $|\alpha| < O(2^{|S|^{6k+\outlen+2}})$.
\end{theorem}
\textbf{Sketch of the Proof} (see the proof in the Appendix). We first show that for every $\alpha \in \mathcal{T}_{out}$ and every $g \in G$ the result of the application $cut\_right$ and $shift$ remains in $\mathcal{T}_{out}$. This property allows us to define the function $\delta: \mathcal{W} \times I \times G \rightarrow{\mathcal{W}}$, where $\mathcal{W} = \binom{S\times\mathcal{T}_{out}}{2} \cup \{\varnothing\}$ is the set of all unordered pairs of (state, timed output tail). In the second step, for TFSM $\mathcal{S}_p$ we define the pairwise automaton (abstraction) $\mathcal{A}_{\mathcal{S}_p} = (\mathcal{D}, (I \times G) ,\tau)$, where $\mathcal{D} = \{W \in 2^\mathcal{W} : |W| \leq \binom{S}{2}\}$, $W_{init} = \{\{(s_1, \epsilon) (s_2, \epsilon) : \{s_1,s_2\} \in \binom{S}{2}\}\}$ and $\tau: \mathcal{D} \times (I \times G) \rightarrow \mathcal{D}$ is the transition relation.

Finally, we show that $\mathcal{S}_p$ has an HS if and only if there exists a path from state $W_{init}$ to any state $W$, where for each $w = \{(s_1, \kappa_1), (s_2, \kappa_2)\} \in W$ we have $\kappa_1 \neq \kappa_2$. Since $$|\mathcal{D}| \leq 2^\mathcal{|W|} =  2^{\binom{|S \times \mathcal{T}_{out}|}{2} + 1} \leq 2^{\binom{|S| \cdot \tlen}{2} + 1},$$ the theorem holds.

\begin{figure}[ht]
  \centering
  \begin{tikzpicture}[node distance = 1.5cm, auto,scale=0.68, every node/.style={scale=1.0}]
\tikzstyle{vecArrow} = [thick, decoration={markings,mark=at position
   1 with {\arrow[semithick]{open triangle 60}}},
   double distance=1.4pt, shorten >= 5.5pt,
   preaction = {decorate},
   postaction = {draw,line width=1.4pt, white,shorten >= 4.5pt}]
\tikzstyle{innerWhite} = [semithick, white,line width=1.4pt, shorten >= 4.5pt]
\tikzstyle{decision} =
        [
                diamond,
                draw,
                fill = green!20,
                text width = 6em,
                text badly centered,
                node distance = 2cm,
                inner sep = 0pt
        ]
\tikzstyle{block} =
        [
                rectangle,
                draw,
                text width = 16em,
                text centered,
                rounded corners,
                minimum height = 2em
        ]
\tikzstyle{block1} =
        [
                rectangle,
                draw,
                text width = 24em,
                text centered,
                rounded corners,
                minimum height = 2em
        ]
\tikzstyle{block2} =
        [
                rectangle,
                draw,
                text width = 34em,
                text centered,
                rounded corners,
                minimum height = 2em
        ]
\tikzstyle{block3} =
        [
                rectangle,
                draw,
                text width = 38em,
                text centered,
                rounded corners,
                minimum height = 2em
        ]
\tikzstyle{block4} =
        [
                rectangle,
                draw,
                text width = 30em,
                text centered,
                rounded corners,
                minimum height = 2em
        ]
\tikzstyle{line} =
        [
                draw,
                -latex'
        ]
\tikzstyle{cloud} =
        [
                draw,
                ellipse,
                fill = red!20,
                node distance = 4cm,
                minimum height = 2em
        ]
\tikzstyle{suite}=[->,>=stealth’,thick,rounded corners=4pt]

        \node [block] (s0) {
            $\overline{(s_0,\{ \varepsilon \}),(s_1,\{ \varepsilon \}),(s_2,\{ \varepsilon \}),(s_3,\{ \varepsilon \})}$
        };
        \node [block1, below of = s0] (s1) {
            $\overline{(s_1,\{(o_1, 2)\},(s_2,\{(o_1, 2)\}),(s_3,\{(o_1, 3)\}),(s_0,\{(o_1, 1)\})}$
        };
        \node [block4, below of = s1] (s2) {
            $\overline{(s_2,\{(o_1,1),(o_1,2)\}),(s_3,\{(o_1,1),(o_1,3)\}),(s_0,\{(o_1,2),(o_1,1)\})}$,
            \dots
        };
        \node [block4, accepting, below of = s2] (s3) {
            $\overline{(s_2,\{(o_1,1),(o_1,2)\}),(s_3,\{(o_1,1),(o_1,3)\}),(s_0,\{(o_1,2),(o_1,1)\})}$,
            \dots
        };
        
        \draw [thick, ->] (s0) -- (s1) node[midway, right] {$(i_1, 1)$};
        \draw [thick, ->] (s1) -- (s2) node[midway, right] {
            $(i_1, 1)$
        };
        \draw [thick, ->] (s2) -- (s3) node[midway, right] {
            $\dots$
        };
\end{tikzpicture}
  \caption{Fragment of the modified truncated successor tree for $\mathcal{B}_4$}
  \label{fig:tst_configuration}
\end{figure}

The proof of Theorem~\ref{theorem:point_tfsm_hs_bound} gives us also the conclusion that if $\frac{\max\{D\}}{\min\{G\}} = poly(|S|)$, then checking whether a given (possibly partial) point-interval deterministic TFSM has a homing sequence is in PSPACE. Indeed, the NPSPACE algorithm would non-deterministically apply $j$-th input of the desired sequence until it reaches the upper bound. Note that each state of abstraction $\mathcal{A}_{\mathcal{S}_p}$ is of the form $W = \{\{(s_1, \kappa_1), (s_2, \kappa_2)\}: s_1, s_2 \in S, \kappa_1, \kappa_2 \in \mathcal{T}_{out}\}$ with $|W| \leq \binom{S}{2}$. Since $\frac{\max\{D\}}{\min\{G\}}$ is polynomial, a timed output tail is also polynomial (see Lemma~\ref{lemma:t_out_len}), therefore we can encode a state of $\mathcal{A}_{\mathcal{S}_p}$ using the polynomial space in terms of $|S|$. In $j$-th iteration, we must store only the state of $\mathcal{A}_{\mathcal{S}_p}$ where we apply the input, the result of that computation and the $j$-th input of the sequence. The Savitch's Theorem~\cite{SAVITCH1970177} concludes the proof, while Theorem~\ref{theorem:pspace_c_pointinterval} establishes the PSPACE-completeness of homing problem for partial TFSMs with only point intervals.

\begin{theorem}
\label{theorem:pspace_c_pointinterval}
Let $\mathcal{S}_p = (S, I ,O, D, G, h_S)$ be a partial point-interval deterministic TFSM, $|h_S| = poly(|S|)$ and $\frac{\max\{D\}}{\min\{G\}} = poly(|S|)$, checking if $\mathcal{S}_p$ has an HS is PSPACE-complete.
\end{theorem}
\begin{proof}
We know that the problem is in PSPACE, so we need only to prove that it is PSPACE-hard. The proof is a reduction of the problem of checking if a given PFA (partial finite automaton) $\mathcal{A}$ is carefully synchronizing \cite{DBLP:journals/mst/Martyugin14}. Let $\mathcal{A} = (Q, \Sigma, \delta)$, we define a point-interval deterministic TFSM $\mathcal{S}_{\mathcal{A}} = (Q, \Sigma,\{o\} ,\{1\}, \{1\}, h_{\delta})$ with $h_\delta = \{(q,a,1,1,\delta(q,a)) : q \in Q \land a \in \Sigma\}$. For any timed sequence $\alpha \in (I \times G)^\ast$ and every pair of states $q_1, q_2 \in Q$ we have $timed\_out_{\mathcal{S}_{\mathcal{A}}}(q_1, \alpha) = timed\_out_{\mathcal{S}_{\mathcal{A}}}(q_2, \alpha)$ (1). Define also for $w = a_1 \ldots a_n$, a sequence $\alpha_w = (a_1, 1)\ldots(a_n, n)$. For every $q$, such that $\delta(q, w) = q'$,  $next\_state_{\mathcal{S}_{\mathcal{A}}}(q, \alpha_w) = q'$ (2). Obviously, from (1), if $w$ is carefully synchronizing for $\mathcal{A}$ (there exists $\bar{q}$ such that $\delta(q,w) = \bar{q}$ for all $q \in Q$), then $\alpha_w$ is homing for $\mathcal{S}_{\mathcal{A}}$. Conversely, if $\alpha = (a_1, 1)\ldots(a_n, n)$ is a homing sequence, then, from (2), there exists $\bar{q}$, such that for every $q \in Q$ $next\_state_{\mathcal{S}_{\mathcal{A}}}(q, \alpha) = \bar{q}$. But this means that $\mathcal{A}$ is carefully synchronized by the word $w_{\alpha} = a_1\ldots a_n$. The reduction is performed in polynomial time, so the result holds.
\end{proof}

\section{Conclusion \& future work}\label{sec:conclusion}
In this paper, we have defined synchronizing and homing sequences for Timed Finite State Machines with output delays and analyzed their properties. We have developed novel approaches for deriving SSs and HSs for TFSMs together with the relevant complexity analysis. Additionally, we have explored the correspondence between these sequences in TFSMs and their FSM abstractions.

This paper opens a number of directions for future work. One important direction is to address the challenge of deriving HSs for TFSMs with arbitrary timed guards. Another problem is how to derive HSs for TFSMs when we cannot observe output response time. Furthermore, it would be valuable to define and investigate the properties of sequences that synchronize (or home) a TFSM not only to a specific state but also to a configuration or location, representing a current state and a combination of concurrently running procedures.
\section*{Acknowledgments} Partial funding for this work was provided by the Erasmus+ program.
\bibliographystyle{elsarticle-harv} 
\bibliography{references}
\appendix
\section{Statement proofs for the Reviewers}
\noindent
\textbf{Theorem~\ref{theorem:correctness_of_hs_synthesis_alg} (Correctness of Algorithm~\ref{alg:hs_synthesis}).} \textit{A weakly-complete deterministic TFSM $\mathcal{S}$ has a homing sequence if and only if the truncated successor tree derived by Algorithm~\ref{alg:hs_synthesis} has a node truncated using Rule 1.
}
\begin{proof}
Given a weakly-complete deterministic TFSM $\mathcal{S} = (S, I, O, G, D, h_{S})$, let $U$ and $V$ denote the minimal left and maximal right boundaries of the timed guards in $G$, respectively. Since the nodes of the TST are labeled with sets of subsets of states, there are at most $2^{2^{|S|}}$ distinct node labels. Due to Rule 2, two nodes with identical labels cannot occur in the same branch. According to the construction of the TST, each node can have at most $(V-U)\cdot |I|$ successors. Therefore, the TST derived by Algorithm~\ref{alg:hs_synthesis} is finite.

$\Leftarrow$ Let $P$ be a node truncated using Rule 1 and $(i_1,\delta_1)(i_2,\delta_2)\dots(i_{\ell},\delta_{\ell})$ labels the path from the root to $P$. We show that $\alpha=(i_1,\delta_1)(i_2,\delta_1+\delta_2)\dots(i_{\ell},\delta_1+\delta_2+\dots+\delta_{\ell})$ is a homing sequence. Assume that it is not true, i.e., there exist states $s$ and $s'$ of $\mathcal{S}$ such that $timed\_out_{\mathcal{S}}(s,\alpha)=timed\_out_{\mathcal{S}}(s',\alpha)$ and $next\_state_{\mathcal{S}}(s,\alpha)\neq next\_state_{\mathcal{S}}(s',\alpha)$. Two options are possible: \textbf{Case 1.} In the path from the root to $P$ there exists a node at level $k$, $k\in \{ 1, \dots, \ell \}$ such that $(i_1,\delta_1)\dots(i_{k-1},\delta_1+\dots+\delta_{k-1})$-successors of $s$ and $s'$ are in the same block, while $(i_1,\delta_1)\dots(i_{k-1},\delta_1+\dots+\delta_{k-1})(i_k,\delta_1+\dots+\delta_{k-1}+\delta_k)$-successors of $s$ and $s'$ are in different blocks; and \textbf{Case 2.} For every $j$, $j\in \{ 1, \dots, \ell \}$ the $(i_1,\delta_1)\dots(i_j,\delta_1+\dots+\delta_{j-1}+\delta_j)$-successors of $s$ and $s'$ are in the same block.

\textbf{Case 1.} Choose the minimal $k$, $k\in \{ 1, \dots, \ell \}$ such that $(i_1,\delta_1)\dots(i_{k-1},\delta_1+\dots+\delta_{k-1})$-successors of $s$ and $s'$ are in the same block and $(i_1,\delta_1)\dots(i_{k-1},\delta_1+\dots+\delta_{k-1})(i_k,\delta_1+\dots+\delta_{k-1}+\delta_k)$-successors of $s$ and $s'$ are in different blocks. The branch labeled with $(i_1,\delta_1)\dots(i_{k},\delta_{k})$ is as follows:$$\{\overline{s, s', \dots} \} \stackrel{i_1,\delta_1}{\longrightarrow} \dots \stackrel{i_{k-1},\delta_{k-1}}{\longrightarrow} \{ \overline{s_{k-1}, s'_{k-1}, \dots}, \dots \} \stackrel{i_k,\delta_k}{\longrightarrow} \{\overline{s_k, \dots}, \overline{s'_k, \dots}, \dots \}.$$
Thus, $timed\_out_{\mathcal{S}}(s_{k-1},(i_k,\delta_k))\neq timed\_out_{\mathcal{S}}(s'_{k-1},(i_k,\delta_k))\Rightarrow$\\ $\Rightarrow \{ \alpha \text{\ is non-integer} \} \Rightarrow 
timed\_out_{\mathcal{S}}(s,(i_1,\delta_1)\dots(i_k,\delta_1+\dots+\delta_k))\neq $\\ $timed\_out_{\mathcal{S}}(s',(i_1,\delta_1)\dots(i_k,\delta_1+\dots+\delta_k)) \Rightarrow \{ \text{the proof of Lemma~\ref{lemma:prolongation_not_hs}} \} \Rightarrow timed\_out_{\mathcal{S}}(s,\alpha)\neq timed\_out_{\mathcal{S}}(s',\alpha)$, it is a contradiction.

\textbf{Case 2.} Assume that for every $j$, $j\in \{ 1, \dots, \ell \}$,  $(i_1,\delta_1)\dots(i_j,\delta_1+\dots+\delta_j)$-successors of $s$ and $s'$ are in the same block. Since $P$ is labeled only with singletons, there exists $k\in \{ 1, \dots, \ell \}$ such that $s_k=s'_k$. Therefore, $next\_state_{\mathcal{S}}(s,(i_1,\delta_1)\dots(i_k,\delta_1+\dots+\delta_k))= next\_state_{\mathcal{S}}(s',(i_1,\delta_1)\dots(i_k,\delta_1+\dots+\delta_k)) \Rightarrow \{ \mathcal{S} \text{ is deterministic} \} \Rightarrow next\_state_{\mathcal{S}}(s,\alpha)= next\_state_{\mathcal{S}}(s',\alpha)$, it is a contradiction.

$\Rightarrow$ Assume that $\mathcal{S}$ has a homing sequence, but the TST derived by Algorithm~\ref{alg:hs_synthesis} does not have any node truncated using Rule 1.\\
Let $\alpha=(i_1,t_1)(i_2,t_2)\dots(i_{\ell},t_{\ell})$ be a shortest HS for TFSM $\mathcal{S}$. Define $\delta_1=\lfloor t_1 \rfloor+2^{-1}$ and $\delta_j=\lfloor t_j-t_{j-1} \rfloor+2^{-j}$ for $j\in \{ 2,\dots,\ell \}$. By this choice of $\delta_1$, $\delta_2$, \dots, $\delta_n$, the sequence $\alpha'=(i_1,\delta_1)(i_2,\delta_1+\delta_2)\dots(i_{\ell},\delta_1+\delta_2+\dots+\delta_{\ell})$ is a non-integer timed input sequence. \\
Since all timed guards are left-closed and right-open, for every state of $\mathcal{S}$ the sequences $\alpha$ and $\alpha'$ activate the same sequence of transitions, therefore $\alpha \sim_{\mathcal{S}} \alpha'$. Thus, $\alpha'$ is also a homing sequence (see Theorem~\ref{theorem:congruent_sequences}). \\
By construction, the TST has the branch labeled with $(i_1,\delta_1)(i_2,\delta_2)\dots(i_{\ell},\delta_{\ell})$ leading to a node $P$. Due to the fact that $\alpha'$ is an HS, $\alpha'$ either splits or merges every pair of states of $\mathcal{S}$, thus $P$ contains only singletons and $P$ is truncated using Rule 1, it is a contradiction.

\end{proof}

\noindent
\textbf{Lemma~\ref{lemma:t_out_len}.}
\textit{Let $\mathcal{S}_p = (S, I, O, G, D, h_S)$ be a (possibly partial) TFSM with only point intervals and $|h_S| = O(|S^k|)$, the following holds:
\begin{enumerate}
    \item $|timed\_out_{\mathcal{S}_p \geq t(\alpha)}(s, \alpha)| \leq \outlen$ for every $s \in S$ and $\alpha \in Dom_{\mathcal{S}_p}(s)$;
    \item $|\mathcal{T}_{out}| = O(|S|^{3k\cdot \outlen})$.
\end{enumerate}}
\begin{proof}
1. We will count how many outputs can occur after time $t(\alpha)$. Observe that any output produced by the input from $\alpha$ applied at time $t < t(\alpha) - \max\{D\}$, must be contained in $timed\_out_{\mathcal{S}_p < t(\alpha)}(s, \alpha)$. The maximum possible number of inputs applied between time $t(\alpha) - \max\{D\}$ and $t(\alpha)$ is exactly $\outlen$. Since the TFSM produces exactly one output for every input, the claim holds.

2. We first will give a precise definition of the set of all pending outputs denoted as $\mathcal{T}_{out}$. Define a set of multisets $$\mathcal{T} = \{ \{(o_1, d'_1),(o_2,d'_2 - t_1),\ldots,(o_k, d'_k - t_{k-1}) \}: (o_i \in O )\land (d'_i \in D) \land (t_1 \in G) \}$$
with \begin{itemize}
    \item $k \leq \outlen$;
    \item $t_{i} = t_{i-1} + g$ for $i \in \{ 2, \dots, k \}$;
    \item $g \in G$.
\end{itemize}

Then we construct set $\mathcal{T}_{out} = \{cut\_right(\kappa,0) : \kappa \in \mathcal{T}\} \cup \{\epsilon\}$.
In other words, set $\mathcal{T}_{out}$ encodes all pending outputs for input sequences, that is, $\mathcal{T}_{out}$ is the set of all possible $shift(timed\_out_{\mathcal{S}_p \geq t(\alpha)}(s, \alpha), -t(\alpha))$ for every $\alpha$ and for every $s$.

Consider $\kappa \in \mathcal{T}$ such that $|\kappa| = k$. Obviously, $\kappa = \{ (o_1, d'_1)( o_t, d'_2 - g'_1) \ldots, (o_k, d'_k - g'_1 - (\sum_{i=2}^{k-1} g'_i)) \}$. Note that we can choose an output of each element in the sequence $\kappa$ in $|O|$ ways, and we can choose a delay in $|D|$ ways. For each next input, we add a guard in one of $|G|$ ways and a delay in $|D|$ ways. So, the number of sequences of length $k$ is equal to $(|O||D|)^k|G|^{k-1}$. \\ 
Thus, $|\mathcal{T}_{out}| - 1 \leq  |\mathcal{T}| \leq \sum_{i=1}^{\outlen}(|O||D||G|)^i = \tlen$. Since $|h_S| = O(|S|^k)$, the claim holds.
\end{proof}

\noindent
\textbf{Theorem~\ref{theorem:point_tfsm_hs_bound}.}
\textit{
Let $\mathcal{S}_p = (S, I, O, G, D, h_S)$ be a (possibly partial) homing TFSM with only point intervals, $|h_S| = O(|S|^k)$ and $\alpha$ be a shortest HS for $\mathcal{S}_p$, it holds that $|\alpha| < O(2^{|S|^{6k+\outlen+2}})$.
}

\begin{proof}
We start with a simple claim: 
\begin{claim}
\label{claim:t_in_W}
If $\kappa \in \mathcal{T}_{out}$, then $cut\_right(shift(\kappa, -g) \cup timed\_out_{\mathcal{S}_p}(s_1,(i,g)), 0)  \in \mathcal{T}_{out}$ for any $s \in S, i \in I, g \in G$.
\end{claim}

\begin{proof}
Note that $timed\_out_{\mathcal{S}_p}(s,(i,g)) = (o,g+d)$ where $o \in O$ and $d \in D$. Denote also $\kappa = (o_1, t_1)\ldots(o_k, t_k)$. According to the definition of $\mathcal{T}_{out}$, we can permutate $\kappa$ to obtain $\kappa' = (o'_1, d'_1)(o'_2,d'_2 - a_1)\ldots(o'_k,d'_n -a_1 - \ldots -a_k)$, where each $a_i = \sum_{j=1}^{l_i} g'_j$, $g'_j \in G$. Now, $\kappa'' = shift(\kappa', -g) \cup (o,d) = (o,d)(o'_1, d'_1 - g)(o'_2,d'_2 - a_1 - g)\ldots(o'_k,d'_n -a_1 - \ldots -a_k - g)$. If $\sum_{i=1}^k l_i = \frac{\max\{D\}}{\min\{G\}}$, then observe (since $d \leq \max\{D\}$, $g \geq \min\{G\}$ and each $a_i = \sum_{j=1}^{l_i} g'_j$) that $d'_n -a_1 - \ldots -a_k - g < 0$, so $cut\_right(\kappa'',0) \in \mathcal{T}_{out}$.
\end{proof}

If $X$ is a set, then, as usual, denote as $\binom{X}{2}$ the set of all pairs of the elements from $X$, and as $2^{X}$ the set of all subsets of $X$. Denote also $\mathcal{W} = \binom{S\times\mathcal{T}_{out}}{2} \cup \{\varnothing\}$. Let $w = \{(s_1,\kappa_1),(s_2,\kappa_2)\}$. Define function $\delta: \mathcal{W} \times I \times G \rightarrow{\mathcal{W}}$ in the following way: 
\begin{enumerate}
 \item if $next\_state_{\mathcal{S}_p}(s_1, (i, g)) = \bot$ or  $next\_state_{\mathcal{S}_p}(s_2, (i,g)) = \bot$, then $\delta(w, i, g)$ is not defined;
 \item \label{def:delta_merge} if $next\_state_{\mathcal{S}_p}(s_1, (i, g)) =  next\_state_{\mathcal{S}_p}(s_2, (i, g))$, then $\delta(w, i, g) = \varnothing$;
 \item if $next\_state_{\mathcal{S}_p}(s_1, (i, g)) = s'_1$ and $next\_state_{\mathcal{S}_p}(s_2, (i, g)) = s'_2$ and $s'_1 \neq s'_2$ and :
   \begin{enumerate}
    \item \label{def:delta_different_less} $cut\_left(shift(\kappa_1, -g) \cup timed\_out_{\mathcal{S}_p}(s_1,(i,g)),0) \neq \\cut\_left(shift(\kappa_2, -g) \cup timed\_out_{\mathcal{S}_p}(s_2,(i,g)),0)$ then $\delta(w, i, g) = \varnothing$;
    \item \label{def:delta_different_not_less}  otherwise $\delta(w, i, g) = \{(s'_1, \kappa'_1)(s'_2, \kappa'_2)\}$ where 
    $$\kappa'_1 = cut\_right(shift(\kappa_1, -g) \cup timed\_out_{\mathcal{S}_p}(s_1,(i,g)), 0)$$   and 
    $$\kappa'_2 = cut\_right(shift(\kappa_2, -g) \cup timed\_out_{\mathcal{S}_p}(s_2,(i,g)), 0)$$ (see Claim \ref{claim:t_in_W});
  \end{enumerate}
  \item \label{def:delta_nothing_bot} $\delta(\varnothing, i, g) = \varnothing $ and $\delta(\bot, i, g) = \bot$.
\end{enumerate}

We can now extend this function $\delta$ to $(I \times G)^\ast$ in a classical manner: 
\begin{itemize}
\item $\delta(w, \epsilon) = w$;
\item $\delta(w, (i,g)\alpha) = \delta(\delta(w,i,g), \alpha)$.
\end{itemize}

Examples of function $\delta$ for the initial parameters $\{(s_0,\epsilon), (s_1, \epsilon)\}$ and $\{(s_0,\epsilon), (s_3, \epsilon)\}$ for machine $\mathcal{B}_4$ are presented in Figures \ref{fig:fig:s0s3_example} and \ref{fig:fig:s0s1_example}. The blue outputs are those added by $timed\_out$ part of function $\delta$, and the red outputs are those shifted by $-1$ (see Point \ref{def:delta_different_not_less}). All outputs for which time is less than $0$ are removed due to function $cut\_right$. Observe also that the sequence of transitions in Fig. \ref{fig:fig:s0s3_example} ends with $\varnothing$ because the condition at Point \ref{def:delta_different_less} is fulfilled.
\begin{figure}[ht]
  \centering

\tikzset{every picture/.style={line width=0.75pt}} 

\begin{tikzpicture}[x=0.75pt,y=0.75pt,yscale=-1,xscale=1]

\draw    (135,37) -- (192,37) ;
\draw [shift={(194,37)}, rotate = 180] [color={rgb, 255:red, 0; green, 0; blue, 0 }  ][line width=0.75]    (10.93,-3.29) .. controls (6.95,-1.4) and (3.31,-0.3) .. (0,0) .. controls (3.31,0.3) and (6.95,1.4) .. (10.93,3.29)   ;
\draw    (298,37) -- (355,37) ;
\draw [shift={(357,37)}, rotate = 180] [color={rgb, 255:red, 0; green, 0; blue, 0 }  ][line width=0.75]    (10.93,-3.29) .. controls (6.95,-1.4) and (3.31,-0.3) .. (0,0) .. controls (3.31,0.3) and (6.95,1.4) .. (10.93,3.29)   ;
\draw    (499,36) -- (556,36) ;
\draw [shift={(558,36)}, rotate = 180] [color={rgb, 255:red, 0; green, 0; blue, 0 }  ][line width=0.75]    (10.93,-3.29) .. controls (6.95,-1.4) and (3.31,-0.3) .. (0,0) .. controls (3.31,0.3) and (6.95,1.4) .. (10.93,3.29)   ;

\draw (90,15) node [anchor=north west][inner sep=0.75pt]   [align=left] {$\displaystyle  \begin{array}{{>{\displaystyle}l}}
( s_{0} ,\epsilon )\\
( s_{3} ,\epsilon )
\end{array}$};
\draw (140,15) node [anchor=north west][inner sep=0.75pt]   [align=left] {$\displaystyle ( i_{1} ,1)$};
\draw (203,16) node [anchor=north west][inner sep=0.75pt]   [align=left] {$\displaystyle  \begin{array}{{>{\displaystyle}l}}
( s_{1} ,[\textcolor[rgb]{0.29,0.56,0.89}{( o_{1} ,2)}])\\
( s_{0} ,[\textcolor[rgb]{0.29,0.56,0.89}{( o_{1} ,1)}])
\end{array}$};
\draw (300,17) node [anchor=north west][inner sep=0.75pt]   [align=left] {$\displaystyle ( i_{1} ,1)$};
\draw (363,16) node [anchor=north west][inner sep=0.75pt]   [align=left] {$\displaystyle  \begin{array}{{>{\displaystyle}l}}
( s_{2} ,[\textcolor[rgb]{0.82,0.01,0.11}{( o_{1} ,1)}\textcolor[rgb]{0.29,0.56,0.89}{( o_{1} ,2)}])\\
( s_{1} ,[\textcolor[rgb]{0.82,0.01,0.11}{( o_{1} ,0)}\textcolor[rgb]{0.29,0.56,0.89}{( o_{1} ,2)}])
\end{array}$};
\draw (501,16) node [anchor=north west][inner sep=0.75pt]   [align=left] {$\displaystyle ( i_{1} ,1)$};
\draw (564,27) node [anchor=north west][inner sep=0.75pt]   [align=left] {$\displaystyle \emptyset $};

\end{tikzpicture}

  \caption{Function $\delta$ for $\{(s_0,\epsilon), (s_3, \epsilon)\}$}
  \label{fig:fig:s0s3_example}
\end{figure}
\begin{figure}[ht]
  \centering

\tikzset{every picture/.style={line width=0.75pt}} 

\begin{tikzpicture}[x=0.75pt,y=0.75pt,yscale=-1,xscale=1]

\draw    (353,30) -- (386,30) ;
\draw [shift={(388,30)}, rotate = 180] [color={rgb, 255:red, 0; green, 0; blue, 0 }  ][line width=0.75]    (10.93,-3.29) .. controls (6.95,-1.4) and (3.31,-0.3) .. (0,0) .. controls (3.31,0.3) and (6.95,1.4) .. (10.93,3.29)   ;
\draw    (479,54) -- (479,115) ;
\draw [shift={(479,117)}, rotate = 270] [color={rgb, 255:red, 0; green, 0; blue, 0 }  ][line width=0.75]    (10.93,-3.29) .. controls (6.95,-1.4) and (3.31,-0.3) .. (0,0) .. controls (3.31,0.3) and (6.95,1.4) .. (10.93,3.29)   ;
\draw    (208,142) -- (179,142) ;
\draw [shift={(177,142)}, rotate = 360] [color={rgb, 255:red, 0; green, 0; blue, 0 }  ][line width=0.75]    (10.93,-3.29) .. controls (6.95,-1.4) and (3.31,-0.3) .. (0,0) .. controls (3.31,0.3) and (6.95,1.4) .. (10.93,3.29)   ;
\draw    (106,116) -- (394.04,57.4) ;
\draw [shift={(396,57)}, rotate = 168.5] [color={rgb, 255:red, 0; green, 0; blue, 0 }  ][line width=0.75]    (10.93,-3.29) .. controls (6.95,-1.4) and (3.31,-0.3) .. (0,0) .. controls (3.31,0.3) and (6.95,1.4) .. (10.93,3.29)   ;
\draw    (53,30) -- (86,30) ;
\draw [shift={(88,30)}, rotate = 180] [color={rgb, 255:red, 0; green, 0; blue, 0 }  ][line width=0.75]    (10.93,-3.29) .. controls (6.95,-1.4) and (3.31,-0.3) .. (0,0) .. controls (3.31,0.3) and (6.95,1.4) .. (10.93,3.29)   ;
\draw    (179,30) -- (212,30) ;
\draw [shift={(214,30)}, rotate = 180] [color={rgb, 255:red, 0; green, 0; blue, 0 }  ][line width=0.75]    (10.93,-3.29) .. controls (6.95,-1.4) and (3.31,-0.3) .. (0,0) .. controls (3.31,0.3) and (6.95,1.4) .. (10.93,3.29)   ;
\draw    (414,141) -- (385,141) ;
\draw [shift={(383,141)}, rotate = 360] [color={rgb, 255:red, 0; green, 0; blue, 0 }  ][line width=0.75]    (10.93,-3.29) .. controls (6.95,-1.4) and (3.31,-0.3) .. (0,0) .. controls (3.31,0.3) and (6.95,1.4) .. (10.93,3.29)   ;

\draw (7,7) node [anchor=north west][inner sep=0.75pt]   [align=left] {$\displaystyle  \begin{array}{{>{\displaystyle}l}}
( s_{0} ,\epsilon )\\
( s_{1} ,\epsilon )
\end{array}$};
\draw (91,8) node [anchor=north west][inner sep=0.75pt]   [align=left] {$\displaystyle  \begin{array}{{>{\displaystyle}l}}
( s_{1} ,[\textcolor[rgb]{0.29,0.56,0.89}{(}\textcolor[rgb]{0.29,0.56,0.89}{o}\textcolor[rgb]{0.29,0.56,0.89}{_{1}}\textcolor[rgb]{0.29,0.56,0.89}{,2}\textcolor[rgb]{0.29,0.56,0.89}{)}])\\
( s_{2} ,[\textcolor[rgb]{0.29,0.56,0.89}{(}\textcolor[rgb]{0.29,0.56,0.89}{o}\textcolor[rgb]{0.29,0.56,0.89}{_{1}}\textcolor[rgb]{0.29,0.56,0.89}{,2}\textcolor[rgb]{0.29,0.56,0.89}{)}])
\end{array}$};
\draw (220,9) node [anchor=north west][inner sep=0.75pt]   [align=left] {$\displaystyle  \begin{array}{{>{\displaystyle}l}}
( s_{2} ,[\textcolor[rgb]{0.82,0.01,0.11}{(}\textcolor[rgb]{0.82,0.01,0.11}{o}\textcolor[rgb]{0.82,0.01,0.11}{_{1}}\textcolor[rgb]{0.82,0.01,0.11}{,1}\textcolor[rgb]{0.82,0.01,0.11}{)}\textcolor[rgb]{0.29,0.56,0.89}{(}\textcolor[rgb]{0.29,0.56,0.89}{o}\textcolor[rgb]{0.29,0.56,0.89}{_{1}}\textcolor[rgb]{0.29,0.56,0.89}{,2}\textcolor[rgb]{0.29,0.56,0.89}{)}])\\
( s_{3} ,[\textcolor[rgb]{0.82,0.01,0.11}{(}\textcolor[rgb]{0.82,0.01,0.11}{o}\textcolor[rgb]{0.82,0.01,0.11}{_{1}}\textcolor[rgb]{0.82,0.01,0.11}{,1}\textcolor[rgb]{0.82,0.01,0.11}{)}\textcolor[rgb]{0.29,0.56,0.89}{(}\textcolor[rgb]{0.29,0.56,0.89}{o}\textcolor[rgb]{0.29,0.56,0.89}{_{1}}\textcolor[rgb]{0.29,0.56,0.89}{,3}\textcolor[rgb]{0.29,0.56,0.89}{)}])
\end{array}$};
\draw (359,15) node [anchor=north west][inner sep=0.75pt]  [font=\tiny] [align=left] {$\displaystyle ( i_{1} ,1)$};
\draw (393,9) node [anchor=north west][inner sep=0.75pt]   [align=left] {$\displaystyle  \begin{array}{{>{\displaystyle}l}}
( s_{3} ,[\textcolor[rgb]{0.82,0.01,0.11}{( o_{1} ,0)(}\textcolor[rgb]{0.82,0.01,0.11}{o}\textcolor[rgb]{0.82,0.01,0.11}{_{1}}\textcolor[rgb]{0.82,0.01,0.11}{,1}\textcolor[rgb]{0.82,0.01,0.11}{)}\textcolor[rgb]{0.29,0.56,0.89}{(}\textcolor[rgb]{0.29,0.56,0.89}{o}\textcolor[rgb]{0.29,0.56,0.89}{_{1}}\textcolor[rgb]{0.29,0.56,0.89}{,3}\textcolor[rgb]{0.29,0.56,0.89}{)}])\\
( s_{0} ,[\textcolor[rgb]{0.82,0.01,0.11}{( o_{1} ,0)}\textcolor[rgb]{0.29,0.56,0.89}{( o_{1} ,1)}\textcolor[rgb]{0.82,0.01,0.11}{(}\textcolor[rgb]{0.82,0.01,0.11}{o}\textcolor[rgb]{0.82,0.01,0.11}{_{1}}\textcolor[rgb]{0.82,0.01,0.11}{,2}\textcolor[rgb]{0.82,0.01,0.11}{)}])
\end{array}$};
\draw (487,71) node [anchor=north west][inner sep=0.75pt]  [font=\tiny] [align=left] {$\displaystyle ( i_{1} ,1)$};
\draw (416,119) node [anchor=north west][inner sep=0.75pt]   [align=left] {$\displaystyle  \begin{array}{{>{\displaystyle}l}}
( s_{0} ,[\textcolor[rgb]{0.82,0.01,0.11}{(}\textcolor[rgb]{0.82,0.01,0.11}{o}\textcolor[rgb]{0.82,0.01,0.11}{_{1}}\textcolor[rgb]{0.82,0.01,0.11}{,0}\textcolor[rgb]{0.82,0.01,0.11}{)}\textcolor[rgb]{0.29,0.56,0.89}{(}\textcolor[rgb]{0.29,0.56,0.89}{o}\textcolor[rgb]{0.29,0.56,0.89}{_{1}}\textcolor[rgb]{0.29,0.56,0.89}{,1}\textcolor[rgb]{0.29,0.56,0.89}{)}\textcolor[rgb]{0.82,0.01,0.11}{(}\textcolor[rgb]{0.82,0.01,0.11}{o}\textcolor[rgb]{0.82,0.01,0.11}{_{1}}\textcolor[rgb]{0.82,0.01,0.11}{,2}\textcolor[rgb]{0.82,0.01,0.11}{)}])\\
( s_{1} ,[\textcolor[rgb]{0.82,0.01,0.11}{(}\textcolor[rgb]{0.82,0.01,0.11}{o}\textcolor[rgb]{0.82,0.01,0.11}{_{1}}\textcolor[rgb]{0.82,0.01,0.11}{,0}\textcolor[rgb]{0.82,0.01,0.11}{)}\textcolor[rgb]{0.82,0.01,0.11}{( o_{1} ,1)}\textcolor[rgb]{0.29,0.56,0.89}{( o_{1} ,2)}])
\end{array}$};
\draw (211,120) node [anchor=north west][inner sep=0.75pt]   [align=left] {$\displaystyle  \begin{array}{{>{\displaystyle}l}}
( s_{1} ,[\textcolor[rgb]{0.82,0.01,0.11}{(}\textcolor[rgb]{0.82,0.01,0.11}{o}\textcolor[rgb]{0.82,0.01,0.11}{_{1}}\textcolor[rgb]{0.82,0.01,0.11}{,0}\textcolor[rgb]{0.82,0.01,0.11}{)}\textcolor[rgb]{0.82,0.01,0.11}{( o_{1} ,1)}\textcolor[rgb]{0.29,0.56,0.89}{( o_{1} ,2)}])\\
( s_{2} ,[\textcolor[rgb]{0.82,0.01,0.11}{(}\textcolor[rgb]{0.82,0.01,0.11}{o}\textcolor[rgb]{0.82,0.01,0.11}{_{1}}\textcolor[rgb]{0.82,0.01,0.11}{,0}\textcolor[rgb]{0.82,0.01,0.11}{)}\textcolor[rgb]{0.82,0.01,0.11}{(}\textcolor[rgb]{0.82,0.01,0.11}{o}\textcolor[rgb]{0.82,0.01,0.11}{_{1}}\textcolor[rgb]{0.82,0.01,0.11}{,1}\textcolor[rgb]{0.82,0.01,0.11}{)}\textcolor[rgb]{0.29,0.56,0.89}{(}\textcolor[rgb]{0.29,0.56,0.89}{o}\textcolor[rgb]{0.29,0.56,0.89}{_{1}}\textcolor[rgb]{0.29,0.56,0.89}{,2}\textcolor[rgb]{0.29,0.56,0.89}{)}])
\end{array}$};
\draw (186,128) node [anchor=north west][inner sep=0.75pt]  [font=\tiny] [align=left] {$\displaystyle ( i_{1} ,1)$};
\draw (6,120) node [anchor=north west][inner sep=0.75pt]   [align=left] {$\displaystyle  \begin{array}{{>{\displaystyle}l}}
( s_{2} ,[\textcolor[rgb]{0.82,0.01,0.11}{(}\textcolor[rgb]{0.82,0.01,0.11}{o}\textcolor[rgb]{0.82,0.01,0.11}{_{1}}\textcolor[rgb]{0.82,0.01,0.11}{,0}\textcolor[rgb]{0.82,0.01,0.11}{)}\textcolor[rgb]{0.82,0.01,0.11}{(}\textcolor[rgb]{0.82,0.01,0.11}{o}\textcolor[rgb]{0.82,0.01,0.11}{_{1}}\textcolor[rgb]{0.82,0.01,0.11}{,1}\textcolor[rgb]{0.82,0.01,0.11}{)}\textcolor[rgb]{0.29,0.56,0.89}{(}\textcolor[rgb]{0.29,0.56,0.89}{o}\textcolor[rgb]{0.29,0.56,0.89}{_{1}}\textcolor[rgb]{0.29,0.56,0.89}{,2}\textcolor[rgb]{0.29,0.56,0.89}{)}])\\
( s_{3} ,[\textcolor[rgb]{0.82,0.01,0.11}{(}\textcolor[rgb]{0.82,0.01,0.11}{o}\textcolor[rgb]{0.82,0.01,0.11}{_{1}}\textcolor[rgb]{0.82,0.01,0.11}{,0}\textcolor[rgb]{0.82,0.01,0.11}{)}\textcolor[rgb]{0.82,0.01,0.11}{(}\textcolor[rgb]{0.82,0.01,0.11}{o}\textcolor[rgb]{0.82,0.01,0.11}{_{1}}\textcolor[rgb]{0.82,0.01,0.11}{,1}\textcolor[rgb]{0.82,0.01,0.11}{)}\textcolor[rgb]{0.29,0.56,0.89}{(}\textcolor[rgb]{0.29,0.56,0.89}{o}\textcolor[rgb]{0.29,0.56,0.89}{_{1}}\textcolor[rgb]{0.29,0.56,0.89}{,3}\textcolor[rgb]{0.29,0.56,0.89}{)}])
\end{array}$};
\draw (227,76) node [anchor=north west][inner sep=0.75pt]  [font=\tiny] [align=left] {$\displaystyle ( i_{1} ,1)$};
\draw (59,15) node [anchor=north west][inner sep=0.75pt]  [font=\tiny] [align=left] {$\displaystyle ( i_{1} ,1)$};
\draw (185,15) node [anchor=north west][inner sep=0.75pt]  [font=\tiny] [align=left] {$\displaystyle ( i_{1} ,1)$};
\draw (393,126) node [anchor=north west][inner sep=0.75pt]  [font=\tiny] [align=left] {$\displaystyle ( i_{1} ,1)$};

\end{tikzpicture}

  \caption{Function $\delta$ for $\{(s_0,\epsilon), (s_1, \epsilon)\}$}
  \label{fig:fig:s0s1_example}
\end{figure}

Observe that, since $\mathcal{S}_p$ is a point-interval machine, any enabled for $\mathcal{S}_p$ sequence $\alpha = (i_1, g'_1)(i_2, g'_1 + g'_2) \ldots (i_k, g'_1 + g'_2\ldots + g'_k)$ can be associated with word $w_\alpha = \\(i_1, g'_1)(i_2, g'_2) \ldots (i_k, g'_k) \in (I \times G)^\ast$.

Let $s_1, s_2 \in S$, we state two claims:

\begin{claim}
\label{claim:pairs_not_varnothing}
Conditions (1) and (2) are equivalent: 
\begin{enumerate}
 \item $next\_state_{\mathcal{S}_p}(s_1,\alpha) \neq next\_state_{\mathcal{S}_p}(s_2,\alpha)$ and \\$timed\_out_{\mathcal{S}_p < t(\alpha)}(s_1, \alpha) = timed\_out_{\mathcal{S}_p < t(\alpha)}(s_2, \alpha)$;
 \item $\delta(\timedpair, w_\alpha) = \{
 (next\_state_{\mathcal{S}_p}(s_1,\alpha), \kappa_1) (next\_state_{\mathcal{S}_p}(s_2,\alpha), \kappa_2)\}$ where $\kappa_1 = shift(timed\_out_{\mathcal{S}_p \geq t(\alpha)}(s_1, \alpha), -t(\alpha))$ and\\ $\kappa_2 = shift(timed\_out_{\mathcal{S}_p \geq t(\alpha)}(s_2, \alpha), -t(\alpha))$.
 \end{enumerate}
\end{claim}

\begin{proof}
$(1) \implies (2)$
The proof follows by induction on the length of $\alpha$. If $|\alpha| = 0$, then the claim holds. Assume it holds for all $\alpha$ shorter or equal to $k$. Consider $\alpha'$ with $|\alpha'| \leq k+1$. We can write $\alpha' = \alpha(i,t(\alpha) + g)$ with $|\alpha| \leq k$. From inductive assumption we know, that $\delta(\timedpair, w_\alpha) = \{
 (next\_state_{\mathcal{S}_p}(s_1,\alpha), \kappa_1) (next\_state_{\mathcal{S}_p}(s_2,\alpha), \kappa_2)\}$ where\\ $\kappa_1 = shift(timed\_out_{\mathcal{S}_p \geq t(\alpha)}(s_1, \alpha), -t(\alpha))$ and\\ $\kappa_2 = shift(timed\_out_{\mathcal{S}_p \geq t(\alpha)}(s_2, \alpha), -t(\alpha))$.\\ Let us calculate $\delta(\delta(\timedpair, w_\alpha), (i,g))$. Since Condition $(1)$ holds for $\alpha(i,t(\alpha) + g)$, we use Point \ref{def:delta_different_not_less} and obtain $\delta(\delta(\timedpair, w_\alpha), (i,g)) = \{(s'_1, \kappa'_1), (s'_2, \kappa'_2)\}$. It is easy to check that $s'_1 = next\_state_{\mathcal{S}_p}(s_1, \alpha(i,t))$ and $s'_2 = next\_state_{\mathcal{S}_p}(s_2, \alpha(i,t))$. To show that\\ $\kappa'_1 = shift(timed\_out_{\mathcal{S}_p \geq t(\alpha(i, t(\alpha) + g))}(s_1, \alpha), -t(\alpha) - g)$ and\\ $\kappa'_2 = shift(timed\_out_{\mathcal{S}_p \geq t(\alpha(i, t(\alpha) + g))}(s_2, t(\alpha) - g)$ it suffices to notice that first we decrease every delay of $\kappa_i$ by $g$, then we add $timed\_out_{\mathcal{S}_p}(s_i,i, g)$ to the end of sequences, then we remove from the sequence all elements with delay less than $0$.

$(2) \implies (1)$

The proof follows by induction on the length of $w_\alpha$. If $|w_\alpha| = 0$, then the claim holds. Assume that it holds for all $w_\alpha$ shorter than or equal to $k$. Consider $w_\alpha'$ with $|w_\alpha'| \leq k+1$. We can write $w_\alpha' = w_\alpha(i,g)$. From inductive assumption, we know that $next\_state_{\mathcal{S}_p}(s_1,\alpha) \neq next\_state_{\mathcal{S}_p}(s_2,\alpha)$ and $timed\_out_{\mathcal{S}_p < t(\alpha)}(s_1, \alpha) = timed\_out_{\mathcal{S}_p < t(\alpha)}(s_2, \alpha)$. Obviously, \\$next\_state_{\mathcal{S}_p}(s_1,\alpha(i,g)) \neq next\_state_{\mathcal{S}_p}(s_2,\alpha(i,g))$. Also, since (2) holds, we know that $\delta(\delta(\timedpair, w_\alpha'), (i,g))$ admits Point \ref{def:delta_different_not_less}, so (see the condition in Point \ref{def:delta_different_less}) we can conclude the proof.
\end{proof}

\begin{claim}
\label{claim:pairs_varnothing}
Conditions (1) and (2) are equivalent: 
\begin{enumerate}
   \item \begin{enumerate}
   \item $next\_state_{\mathcal{S}_p}(s_1, \alpha) = next\_state_{\mathcal{S}_p}(s_2, \alpha)$ or 
   \item $next\_state_{\mathcal{S}_p}(s_1,\alpha) \neq next\_state_{\mathcal{S}_p}(s_2,\alpha)$ and\\ $timed\_out_{\mathcal{S}_p < t(\alpha)}(s_1, \alpha) \neq timed\_out_{\mathcal{S}_p < t(\alpha)}(s_2, \alpha)$;
   \end{enumerate}
   \item $\delta(\timedpair, w_\alpha) = \varnothing$.
\end{enumerate}
\end{claim}

\begin{proof}

$(1) \implies (2)$

We will show that $(1a) \implies (2)$ or $(1b) \implies (2)$.
If we assume $(1a)$, then we know that there exists $a'(i,t)$, a prefix of $\alpha$ such that $next\_state_{\mathcal{S}_p}(s_1, \alpha'(i,t))$ $ = next\_state_{\mathcal{S}_p}(s_2, \alpha'(i,t))$ and  $next\_state_{\mathcal{S}_p}(s_1, \alpha') \neq next\_state_{\mathcal{S}_p}(s_2, \alpha')$. From that, we know that \\$\delta(\timedpair, w_{\alpha'}) = \{(s'_1, \kappa'_1),(s'_2, \kappa'_2)\}$ and $\delta(\timedpair, w_{\alpha'}(i, t-t_{\alpha'})) = \varnothing$.

If we assume $(1b)$, then we know that there exists $a'(i,t)$, a prefix of $\alpha$ such that:
\begin{itemize}
\item $next\_state_{\mathcal{S}_p}(s_1, \alpha'(i,t)) \neq next\_state_{\mathcal{S}_p}(s_2, \alpha'(i,t))$;
\item $next\_state_{\mathcal{S}_p}(s_1, \alpha') \neq next\_state_{\mathcal{S}_p}(s_2, \alpha')$;
\item $timed\_out_{\mathcal{S}_p < t(\alpha')}(s_1, \alpha') = timed\_out_{\mathcal{S}_p < t(\alpha')}(s_2, \alpha')$;
\item \label{item:important}$timed\_out_{\mathcal{S}_p < t(\alpha'(i,t))}(s_1, \alpha'(i,t)) \neq timed\_out_{\mathcal{S}_p < t(\alpha'(i,t))}(s_2, \alpha'(i,t))$.

\end{itemize}
Using Claim \ref{claim:pairs_not_varnothing} we know that $\delta(\timedpair,w_{\alpha'}) = \{
 (s'_1, \kappa_1) (s'_2, \kappa_2)\}$ where: 
 \begin{itemize}
 \item $\kappa_1 = shift(timed\_out_{\mathcal{S}_p \geq t(\alpha')}(s_1, \alpha'), -t(\alpha'))$;
 \item $\kappa_2 = shift(timed\_out_{\mathcal{S}_p \geq t(\alpha')}(s_2, \alpha'), -t(\alpha'))$;
 \item $s'_1 = next\_state_{\mathcal{S}_p}(s_1,\alpha')$;
 \item$s'_2 = next\_state_{\mathcal{S}_p}(s_2,\alpha')$.
 \end{itemize} But also it is easy to check that $$cut\_left(shift(\kappa_1, t - t(\alpha'(i,t))) \cup timed\_out_{\mathcal{S}_p}(s'_1,(i,t - t(\alpha'(i,t)))),0) \neq$$ $$ cut\_left(shift(\kappa_2, t - t(\alpha'(i,t))) \cup timed\_out_{\mathcal{S}_p}(s'_2,(i,t - t(\alpha'(i,t)))),0),$$ so we apply Point \ref{def:delta_different_less} of the definition of $\delta$.

Using now induction with Point \ref{def:delta_nothing_bot} ends that case.

$(2) \implies (1)$
Assume $(2)$. Then we know that there is prefix $w_\alpha'(i,g)$ of $w_\alpha$ such that $\delta(\timedpair,w_\alpha'(i,g)) = \varnothing$ and $\delta(\timedpair, w_\alpha') \neq \varnothing$. But that implies  $\delta(\delta(\timedpair,w_\alpha'),(i,g))$ admits Point \ref{def:delta_merge} or Point \ref{def:delta_different_less} of the definition of $\delta$. It is easy to check that either $(1a)$ or $(1b)$ holds. 
\noindent

\end{proof}

Let $\mathcal{D} = \{W \in 2^\mathcal{W} : |W| \leq \binom{S}{2}\}$ and note $W_{init} = \{\{(s_1, \epsilon) (s_2, \epsilon) : \{s_1,s_2\} \in \binom{S}{2}\}\}$. Obviously $W_{init} \in \mathcal{D}$. We construct, for a given $\mathcal{S}_p$, an automaton (abstraction) $\mathcal{A}_{\mathcal{S}_p} = (\mathcal{D}, (I \times G) ,\tau)$ where $\tau: \mathcal{D} \times (I \times G) \rightarrow \mathcal{D}$, and  $\tau(W, (i,g)) = \bigcup_{w\in W}\delta(w, (i,g))$ if $\delta(w, (i,g)) \neq \bot$ for all $w \in W$, otherwise $\delta(w, (i,g)) = \bot$. Since $\mathcal{S}_p$ is deterministic, we know that for all $w_{\alpha} \in (I \times G)^\ast$, it holds $|\tau(W_{init}, w_\alpha)| \leq |W_{init}| = \binom{S}{2}$ so the function $\tau$ is well defined (the image remains in $\mathcal{D}$). An example of the first few states of the automaton $\mathcal{A}_{\mathcal{B}_4}$ is shown in Fig. \ref{fig:fig:a_b_4_example}. Initial state of that automaton is $W_{init}$. Observe that the fourth state encodes only three pairs. Indeed, three pairs of the third state (those with $(o_1, 0)$ in the second coordinate) fulfill the condition \ref{def:delta_different_less} of the definition of $\delta$, which can be easily checked. 

\begin{figure}[ht]
  \centering

\tikzset{every picture/.style={line width=0.75pt}} 

\begin{tikzpicture}[x=0.75pt,y=0.75pt,yscale=-0.91,xscale=0.91]

\draw    (115,91) -- (148,91) ;
\draw [shift={(150,91)}, rotate = 180] [color={rgb, 255:red, 0; green, 0; blue, 0 }  ][line width=0.75]    (10.93,-3.29) .. controls (6.95,-1.4) and (3.31,-0.3) .. (0,0) .. controls (3.31,0.3) and (6.95,1.4) .. (10.93,3.29)   ;
\draw   (1,40.8) .. controls (1,28.76) and (10.76,19) .. (22.8,19) -- (88.2,19) .. controls (100.24,19) and (110,28.76) .. (110,40.8) -- (110,128.2) .. controls (110,140.24) and (100.24,150) .. (88.2,150) -- (22.8,150) .. controls (10.76,150) and (1,140.24) .. (1,128.2) -- cycle ;
\draw   (153,45.4) .. controls (153,30.82) and (164.82,19) .. (179.4,19) -- (289.6,19) .. controls (304.18,19) and (316,30.82) .. (316,45.4) -- (316,124.6) .. controls (316,139.18) and (304.18,151) .. (289.6,151) -- (179.4,151) .. controls (164.82,151) and (153,139.18) .. (153,124.6) -- cycle ;
\draw   (352,47) .. controls (352,32.64) and (363.64,21) .. (378,21) -- (569,21) .. controls (583.36,21) and (595,32.64) .. (595,47) -- (595,125) .. controls (595,139.36) and (583.36,151) .. (569,151) -- (378,151) .. controls (363.64,151) and (352,139.36) .. (352,125) -- cycle ;
\draw    (318,90) -- (351,90) ;
\draw [shift={(353,90)}, rotate = 180] [color={rgb, 255:red, 0; green, 0; blue, 0 }  ][line width=0.75]    (10.93,-3.29) .. controls (6.95,-1.4) and (3.31,-0.3) .. (0,0) .. controls (3.31,0.3) and (6.95,1.4) .. (10.93,3.29)   ;
\draw   (136,253.2) .. controls (136,245.36) and (142.36,239) .. (150.2,239) -- (441.8,239) .. controls (449.64,239) and (456,245.36) .. (456,253.2) -- (456,295.8) .. controls (456,303.64) and (449.64,310) .. (441.8,310) -- (150.2,310) .. controls (142.36,310) and (136,303.64) .. (136,295.8) -- cycle ;
\draw    (472,151) -- (302.77,240.07) ;
\draw [shift={(301,241)}, rotate = 332.24] [color={rgb, 255:red, 0; green, 0; blue, 0 }  ][line width=0.75]    (10.93,-3.29) .. controls (6.95,-1.4) and (3.31,-0.3) .. (0,0) .. controls (3.31,0.3) and (6.95,1.4) .. (10.93,3.29)   ;
\draw    (133,279) -- (81,279) ;
\draw [shift={(79,279)}, rotate = 360] [color={rgb, 255:red, 0; green, 0; blue, 0 }  ][line width=0.75]    (10.93,-3.29) .. controls (6.95,-1.4) and (3.31,-0.3) .. (0,0) .. controls (3.31,0.3) and (6.95,1.4) .. (10.93,3.29)   ;

\draw (6,29) node [anchor=north west][inner sep=0.75pt]  [font=\scriptsize] [align=left] {$\displaystyle  \begin{array}{{>{\displaystyle}l}}
\{( s_{0} ,\epsilon ) ,( s_{1} ,\epsilon )\}\\
\{( s_{0} ,\epsilon ) ,( s_{2} ,\epsilon )\}\\
\{( s_{0} ,\epsilon ) ,( s_{3} ,\epsilon )\}\\
\{( s_{1} ,\epsilon ) ,( s_{2} ,\epsilon )\}\\
\{( s_{1} ,\epsilon ) ,( s_{3} ,\epsilon )\}\\
\{( s_{2} ,\epsilon ) ,( s_{3} ,\epsilon )\}
\end{array}$};
\draw (115,75) node [anchor=north west][inner sep=0.75pt]  [font=\tiny] [align=left] {$\displaystyle ( i_{1} ,1)$};
\draw (155,30) node [anchor=north west][inner sep=0.75pt]  [font=\scriptsize] [align=left] {$\displaystyle  \begin{array}{{>{\displaystyle}l}}
\{( s_{1} ,( o_{1} ,2)) ,( s_{2} ,( o_{1} ,2))\}\\
\{( s_{1} ,( o_{1} ,2)) ,( s_{3} ,( o_{1} ,3))\}\\
\{( s_{1} ,( o_{1} ,2)) ,( s_{0} ,( o_{1} ,1))\}\\
\{( s_{2} ,( o_{1} ,2)) ,( s_{3} ,( o_{1} ,3))\}\\
\{( s_{2} ,( o_{1} ,2)) ,( s_{0} ,( o_{1} ,1))\}\\
\{( s_{3} ,( o_{1} ,3)) ,( s_{0} ,( o_{1} ,1))\}
\end{array}$};
\draw (356,30) node [anchor=north west][inner sep=0.75pt]  [font=\scriptsize] [align=left] {$\displaystyle  \begin{array}{{>{\displaystyle}l}}
\{( s_{2} ,( o_{1} ,1)( o_{1} ,2)) ,( s_{3} ,( o_{1} ,1)( o_{1} ,3))\}\\
\{( s_{2} ,( o_{1} ,1)( o_{1} ,2)) ,( s_{0} ,( o_{1} ,2)( o_{1} ,1))\}\\
\{( s_{2} ,( o_{1} ,1)( o_{1} ,2)) ,( s_{1} ,( o_{1} ,0)( o_{1} ,2))\}\\
\{( s_{3} ,( o_{1} ,1)( o_{1} ,3)) ,( s_{0} ,( o_{1} ,2)( o_{1} ,1))\}\\
\{( s_{3} ,( o_{1} ,1)( o_{1} ,3)) ,( s_{1} ,( o_{1} ,0)( o_{1} ,2))\}\\
\{( s_{0} ,( o_{1} ,2)( o_{1} ,1)) ,( s_{1} ,( o_{1} ,0)( o_{1} ,2))\}
\end{array}$};
\draw (317,75) node [anchor=north west][inner sep=0.75pt]  [font=\tiny] [align=left] {$\displaystyle ( i_{1} ,1)$};
\draw (141,248) node [anchor=north west][inner sep=0.75pt]  [font=\scriptsize] [align=left] {$\displaystyle  \begin{array}{{>{\displaystyle}l}}
\{( s_{3} ,( o_{1} ,0)( o_{1} ,1)( o_{1} ,3)) ,( s_{0} ,( o_{1} ,0)( o_{1} ,2)( o_{1} ,1))\}\\
\{( s_{3} ,( o_{1} ,0)( o_{1} ,1)( o_{1} ,3)) ,( s_{1} ,( o_{1} ,1)( o_{1} ,0)( o_{1} ,2))\}\\
\{( s_{0} ,( o_{1} ,0)( o_{1} ,2)( o_{1} ,1)) ,( s_{1} ,( o_{1} ,1)( o_{1} ,0)( o_{1} ,2))\}
\end{array}$};
\draw (372,181) node [anchor=north west][inner sep=0.75pt]  [font=\tiny] [below, align=left] {$\displaystyle ( i_{1} ,1)$};
\draw (95,267) node [anchor=north west][inner sep=0.75pt]  [font=\tiny] [align=left] {$\displaystyle ( i_{1} ,1)$};
\draw (50,267) node [anchor=north west][inner sep=0.75pt]   [align=left] {...};

\end{tikzpicture}

  \caption{The automaton $\mathcal{A}_{\mathcal{B}_4}$}
  \label{fig:fig:a_b_4_example}
\end{figure}

Observe that $\alpha$ is a homing sequence for $\mathcal{S}_p$ if and only if $w_\alpha$ labels a path from state $W_{init}$ to any state $W$ where for each $w = \{(s_1, \kappa_1), (s_2, \kappa_2)\} \in W$ we have $\kappa_1 \neq \kappa_2$ (then by Claims \ref{claim:pairs_not_varnothing} and \ref{claim:pairs_varnothing} we know that each pair is either merged or split with the corresponding output response). Since $$|\mathcal{D}| \leq 2^\mathcal{|W|} =  2^{\binom{|S \times \mathcal{T}_{out}|}{2} + 1} \leq 2^{\binom{|S| \cdot \tlen}{2} + 1}$$ (see Lemma~\ref{lemma:t_out_len}), the theorem holds. 
\end{proof}

\end{document}